\documentclass[11pt,letterpaper]{article}

\pdfoutput=1

\usepackage{amsmath,amsthm,amssymb} %
\usepackage[english]{babel} %
\usepackage[style=ieee,sorting=nty,url=false,isbn=false]{biblatex} %
\usepackage{caption} %
\usepackage{color} %
\usepackage{csquotes} %
\usepackage{dsfont} %
\usepackage{fourier} 
\usepackage{framed} %
\usepackage[hmargin=1.25in,vmargin=1in]{geometry} %
\usepackage{graphicx} %
\usepackage[breaklinks]{hyperref} %
\usepackage{soul} %
\usepackage{subcaption} %
\usepackage{suffix} 
\usepackage{xspace} %

\hypersetup{colorlinks=true, linkcolor=blue, citecolor=magenta}

\renewenvironment{framed}[1][0.9\linewidth]{
  \MakeFramed{\setlength{\hsize}{#1}\FrameRestore}}{\endMakeFramed}

\setul{1ex}{.5pt}

\newtheorem{theorem}{Theorem} %
\newtheorem{lemma}[theorem]{Lemma} %
\newtheorem{corollary}[theorem]{Corollary} %
\newtheorem{definition}[theorem]{Definition} %

\newcommand{\ket}[1]{\ensuremath{|#1\rangle}} %
\newcommand{\bra}[1]{\ensuremath{\langle#1|}} %
\newcommand{\bigket}[1]{\ensuremath{\left|#1\right\rangle}} %
\newcommand{\bigbra}[1]{\ensuremath{\left\langle#1\right|}} %
\newcommand{\op}[2]{\ensuremath{\left|#1\right\rangle
    \left\langle#2\right|}} %
\newcommand{\ip}[2]{\ensuremath{\left\langle#1,#2\right\rangle}} %
\newcommand{\norm}[1]{\ensuremath{\left\lVert #1 \right\rVert}} %
\newcommand{\abs}[1]{\ensuremath{\left\lvert #1 \right\rvert}} %
\newcommand{\complex}{\mathbb{C}} %
\newcommand{\real}{\mathbb{R}} %
\newcommand{\class}[1]{\ensuremath{\text{\bf #1}}} 

\DeclareMathOperator{\CNOT}{CNOT} %
\DeclareMathOperator{\SWAP}{SWAP} %

\DeclareMathOperator{\Density}{D} %
\DeclareMathOperator{\Herm}{Herm} %
\DeclareMathOperator{\Lin}{L} %
\DeclareMathOperator{\Pos}{Pos} %

\DeclareMathOperator{\poly}{poly} %
\DeclareMathOperator{\tr}{tr} %
\newcommand{\E}{\mathop{\mathbb{E}}\displaylimits} 
\renewcommand{\Re}{\operatorname{Re}} %

\newcommand{\vq}{{\vec{q}}} %
\newcommand{\dtr}{\ensuremath{D_{\tr}}} %

\def\B{\mathcal{B}} %
\def\H{\mathcal{H}} %
\def\P{\mathcal{P}} %
\def\S{\mathcal{S}} %
\def\T{\mathcal{T}} %
\def\Stab{\mathfrak{S}} %

\begin{document}


\title{\bf\Large Classical Verification of Quantum Proofs}

\author{Zhengfeng Ji\\[.2em]
  \it\normalsize Institute for Quantum Computing and School of Computer Science,\\
  \it\normalsize University of Waterloo, Waterloo, Ontario, Canada\\[.1em]
  \it\normalsize State Key Laboratory of Computer Science, Institute of Software,\\
  \it\normalsize Chinese Academy of Sciences, Beijing, China}

\date{}

\maketitle

\begin{abstract}
  We present a classical interactive protocol that verifies the
  validity of a quantum witness state for the local Hamiltonian
  problem. It follows from this protocol that approximating the
  non-local value of a multi-player one-round game to inverse
  polynomial precision is \class{QMA}-hard. Our work makes an
  interesting connection between the theory of
  \class{QMA}-completeness and Hamiltonian complexity on one hand and
  the study of non-local games and Bell inequalities on the other.
\end{abstract}

\section{Introduction}

The concept of efficient proof verification is of fundamental
importance to the theory of computation. The complexity class
\class{NP} abstracts the notion of checking written proof strings by a
polynomial-time deterministic verifier. It is hard to overstate the
importance of \class{NP} and \class{NP}-completeness
theory~\cite{Coo71,Lev73,Kar72} to the development of theoretical
computer science in the past several decades.

Interactive models of proof verification was proposed and studied by
Babai~\cite{Bab85} and Goldwasser, Micali, and Rackoff~\cite{GMR85}.
It is generalized to the multiple-prover setting by Ben-Or,
Goldwasser, Kilian and Wigderson~\cite{BGKW88}. The efforts to
understand these interactive proof systems opened the door to a series
of breakthroughs in computational complexity
theory~(e.g.,~\cite{LFKN92,Sha92,BFL90,AS98,ALM+98}).

Of particular interest to this work is the multi-player one-round
game, a game theoretical model for interactive proof verification. In
this model, the verifier samples a tuple of questions and sends them
to the players and receives answers back from the players in one
round. The players may agree on a particular strategy, but are
otherwise not allowed to communicate with each other during the game.
The verifier decides whether to accept or to reject based on the
questions and answers. A method called oracularization~\cite{FRS94}
that enforces the functional behavior of the players provides a
multi-player-game characterization of \class{NP}, in which the
questions are bit strings of length logarithmic in input length and
answers are strings of constant number of bits. We emphasize that this
is an exponential save in the number of bits communicated from the
provers to the verifier compared with the standard way of proof
communication in which the prover send the full proof string to the
verifier. This is one of the reasons behind the unexpected power of
multi-prover interactive proof systems and the possibility of
achieving probabilistically checkable proofs~\cite{AS98,ALM+98,Din07}.

The study of quantum variants of proof verification systems provides
both a fruitful way of understanding proof verification in the context
of quantum computing and an insightful angle from which to view unique
quantum phenomena such as entanglement and non-locality.

A quantum analog of efficient verification of written proofs was
proposed by Kitaev~\cite{Kit99,KSV02,AN02}. In this generalization, a
quantum witness state plays the role of the written proof and a
polynomial-time quantum computer checks whether the witness state is
valid for the input. Kitaev introduces the class \class{QMA} of
problems that admit efficient verifiable quantum proofs. He also
establishes the quantum analog of the Cook-Levin theorem by showing
that the local Hamiltonian problem, the natural quantum version of the
constraint satisfaction problems, is complete for \class{QMA}. The
study of local Hamiltonian problems, the structure of entanglement in
the ground states of local Hamiltonians, and the quantum PCP
conjecture~(see e.g.,~\cite{AAV13}) form a new research direction
called Hamiltonian complexity~\cite{Osb12,GHLS14}.

Quantum interactive proof systems, a model in which a quantum
polynomial-time verifier exchanges quantum messages with an
all-powerful prover, were first studied by Kitaev and
Watrous~\cite{KW00,Wat99}. It is now known that the class of languages
expressible by such proof systems, \class{QIP}, is the same as its
classical counterpart, \class{IP}~\cite{JJUW11}.

In quantum multi-prover interactive proof systems, shared entanglement
among the provers plays an essential role. It is known that, without
shared entanglement, or with limited amount of entanglement, the
collection of languages that have quantum multi-prover interactive
proof systems, \class{QMIP}, equals to the classical counterpart,
\class{MIP}~\cite{KM03} (and, hence, also equals to
\class{NEXP}~\cite{BFL90}). The classes $\class{QMIP}^*$ and
$\class{MIP}^*$, corresponding to the collections of languages that
have entangled multi-prover interactive proofs with quantum and
classical messages respectively, are also known to be the
same~\cite{RUV13}. There is a recent evidence showing that entangled
provers may be more powerful than classical provers~\cite{FV15}. But a
full understanding of these two complexity classes is still out of
reach.

In this work, we focus on the multi-prover one-round games with
entangled provers. For a multi-player one-round game, its classical
value is the maximum acceptance probability that classical strategies
can achieve and its non-local value is the maximum acceptance
probability that entangled players can achieve. In general, the
non-local value can be strictly larger than the classical value of a
multi-player game. It is pointed out in~\cite{CHTW04} that this may
cause problems in multi-prover interactive proof systems as provers
with shared entanglement may break the soundness condition of a
classically sound protocol. One striking example is given by the
so-called magic square game~\cite{Mer90,Per90}, which has non-local
value one even though it corresponds to a system of constraints with
no classical solution~\cite{CHTW04}. Strong evidences are also given
in that paper that the entanglement between the players may indeed
weaken the power of two-player XOR games.

Several methods have been proposed to control the cheating ability of
entangled provers and recover soundness in certain cases. It is proved
that approximating the non-local value of a multi-player game to
inverse-polynomial precision is \class{NP}-hard~\cite{KKM+08,IKM09}.
Several natural problems arise from the study of non-locality,
including the binary constraint system game~\cite{CM12}, the quantum
coloring game~\cite{CMN+07} and the game corresponding to the
Kochen-Specker sets~\cite{KS67}, are shown to be \class{NP}-hard
in~\cite{Ji13}. By proving that the multi-linearity test~\cite{BFL90}
is sound against entangled provers, Ito and Vidick proved the
containment of \class{NEXP} in $\class{MIP}^*$~\cite{IV12}. This was
later improved to the result that three-player XOR games are
\class{NP}-hard to approximate even to constant
precision~\cite{Vid13}.

However, after many years of research, there is still no upper bound
on the expressive power of multi-player interactive proofs with
entangled provers. This leaves open the possibility that even a classical
verifier may be able to design a protocol in which only provers with
entanglement can follow and convince the verifier of statements that
classical provers cannot prove.

In this paper, we present a multi-player one-round game for the local
Hamiltonian problem in which the verifier is classical and samples
questions of logarithmic size and expects answers of constant size. As
a corollary, the problem of approximating the non-local value of a
multi-prover one-round game is \class{QMA}-hard, improving
the \class{NP}-hardness results in previous works. It makes an
interesting connection between the theory of \class{QMA}-completeness
and Hamiltonian complexity on one hand, and the study of non-local
games and Bell-inequalities on the other. This also provides an
example that a classical verifier can design protocols making
essential use of the shared entanglement between the provers and
expect them to do things that is impossible for provers without shared
entanglement unless $\class{NP} = \class{QMA}$.

Our work can be thought of as a de-quantization of the
Fitzsimons-Vidick protocol~\cite{FV15} of both the verifier and the
messages. The verifier communicates with multiple entangled provers
and delegates the quantum verification procedure to the provers. In
this sense, this work is also relevant to the developments in the
delegation of quantum computation and blind quantum
computing~\cite{BFK09,ABE10,RUV13,FK12}. The previous works usually
use a cluster state~\cite{RBB03} or EPR states and teleportation to
encode quantum computation, while our approach has the additional
freedom to encode quantum data directly among the provers. This allows
us to go from the delegation of quantum computation to the delegation
of quantum proof verification.

The main result of this paper is stated in the following theorem.

\begin{theorem}
  \label{thm:main}
  For integer $r\ge 4$, any promise problem
  $L=(L_{\text{yes}}, L_{\text{no}})$ in \class{QMA}, and any instance
  $x$ of the problem, there exists an $r$-player one-round game and
  real numbers $c,s\in [0,1]$, $c-s \ge 1/\poly(\abs{x})$ such that
  \begin{enumerate}
  \item The questions are classical bit strings of length
    $O(\log(\abs{x}))$.
  \item The answers are classical bit strings of length $O(1)$.
  \item If $x\in L_{\text{yes}}$, then the non-local value of the game
    is at least $c$.
  \item If $x\in L_{\text{no}}$, then the non-local value of the game
    is at most $s$.
  \end{enumerate}
\end{theorem}

A direct corollary is that approximating the non-local value of a
multi-player game is \class{QMA}-hard.

\begin{corollary}
  Given a multi-player one-round game in which the questions are
  strings of $O(\log n)$ bits and answers are of strings of $O(1)$
  bits, it is \class{QMA}-hard to approximate the non-local value of
  the game to inverse polynomial precision.
\end{corollary}

The same problem for the classical value is obviously in \class{NP}.
This means that the non-local value of multi-player one-round games is
strictly harder to approximate than the classical value unless
$\class{NP} = \class{QMA}$.

Our result has the following consequence for the multi-prover
interactive proofs with entangled provers by scaling up the problem
size. It is a slight improvement of the results obtained
in~\cite{FV15}. Let $\class{QMA}_{\class{EXP}}$ be the collection of
problems that have quantum witnesses of exponentially many qubits
verifiable by a quantum exponential-time machine, and let
$\class{MIP}^*(r,t,c,s)$ be the class of languages that have
$r$-prover, $t$-round interactive proofs with a classical
polynomial-time verifier, entangled provers, and completeness $c$,
soundness $s$.

\begin{corollary}
  \label{cor:MIP*}
  For some choices of completeness and soundness $c,s$ and polynomial
  $p(n)$, with $c-s=\Omega(\exp(-p(n)))$,
  \begin{equation*}
    \class{QMA}_{\class{EXP}} \subseteq \class{MIP}^{*}(4,1,c,s),
  \end{equation*}
  and hence
  \begin{equation*}
    \class{MIP} \subsetneq \class{MIP}^{*}(4,1,c,s)
  \end{equation*}
  unless $\class{NEXP}=\class{QMA}_{\class{EXP}}$.
\end{corollary}

\subsection{Techniques and Proof Overview}

The main technical difficulty we face is how a classical verifier can
check the quantum witness state distributed among a number of players.
In a one-round game, the only thing that the classical verifier can
collect is some information about the conditional distributions
$\Pr(a|q)$ for all possible questions $q$ and answers $a$. Consider
the situation of remote state certification, in which two players $A$
and $B$ share a quantum state $\rho_{AB}$ and want to convince the
verifier of this fact. If the state $\rho_{AB}$ is an EPR state
$(\ket{00}+\ket{11}) / \sqrt{2}$, this is possible in some sense by
the verifier playing the CHSH game~\cite{CHSH69} with the players. The
rigidity of the CHSH game~\cite{MYS12,RUV13} implies that if the
players win the CHSH game with almost optimal probability, then the
state is close to the EPR state up to local isometries. If the state
is mixed, however, the situation becomes problematic in a very strong
sense. Suppose the two players want to prove that $\rho_{AB}$ is the
Werner state~\cite{Wer89}
\begin{equation*}
  \rho_{\text{W}}(\phi) = \frac{(d-\phi)I + (d\phi-1) \SWAP_d}{d^3-d},
\end{equation*}
where $\SWAP_d$ is unitary operator satisfying
$\SWAP_d \ket{i,j} = \ket{j,i}$ for all $0\le i,j \le d-1$. It is
shown by Werner~\cite{Wer89} that, for some choices of $\phi$, the
state is an entangled state, but any prescribed local measurement
setting on $A$ and $B$ performed on the state $\rho_{\text{W}}$
produces distributions $\Pr(a|q)$ that have local hidden variable
models. That is, the distribution can be exactly reproduced by two
classical players with shared randomness and no shared entangled
states whatsoever!

It is natural to consider methods from the study of device-independent
quantum information processing or self-testing quantum
devices~(e.g.,~\cite{MY98,DMMS00,MS12,McK14}). For example, such ideas
have successful applications in achieving the classical command of
quantum system as shown in~\cite{RUV13}. A key ingredient behind such
device independent setting is the rigidity of non-local games such as
the CHSH game. By the definition of rigidity, however, the players
will essentially share a specific entangled state, such as the EPR
state or the GHZ state $(\ket{000}+\ket{111}) / \sqrt{2}$, and perform
prescribed measurements on the state. This seems contradictory to what
we need here---the ability to store the quantum witness state
distributed among the players. The quantum witness state is usually an
entangled state with complex structures that are far way from what EPR
or GHZ states can represent.

Our solution that resolves the above mentioned difficulties is to
encode the quantum witness state by certain stabilizer code and play a
new game, which we call the stabilizer game, defined by the
stabilizer. We prove a rigidity theorem for the stabilizer game which
roughly states that the only way for the players to win the game with
high probability is to share a correctly encoded state of the
stabilizer code, perform measurements according to the measurement
specifications given in the questions, and respond with the
measurement outcome. That is, the stabilizer game we construct
provides a device-independent verification of the encoding of the
corresponding stabilizer code. Having both the rigidity property and
the ability to encode quantum data, the stabilizer game serves as an
essential tool for our work and may find other applications in
device-independent quantum information processing. For example, in the
remote state certification problem discussed above, although it is
impossible for the players to certify the Werner state
$\rho_{\text{W}}$, the stabilizer games provide a way to certify an
\emph{encoded} Werner state using a stabilizer code.

In~\cite{FV15}, quantum error correcting codes are also employed in an
essential way. The intuition is that, by the quantum error correction
property, there will be only one qubit of the player that can pass the
encoding check of the error correcting code once the rest of the
players respond with the correct qubit. In our case, the stabilizer
codes have the same effect but also have an important additional use.
Namely, they are responsible for enforcing the players to measure
their systems according to the measurement specifications they
receive. Instead of using the decoding circuit of the code, we measure
the logical $X$, $Z$ operators on the encoded state. As a result, our
proof directly applies to quantum error detecting codes as well.

The proof of the rigidity for the stabilizer games consists of two
steps. In the first step, an idea motivated by the CHSH game is
employed and we introduce the special-player stabilizer game. Using
similar techniques for proving the rigidity of the CHSH game, we
establish a pair of anti-commuting reflections in the strategy for the
special player. Then, a quantity called consistency is used to promote
this partial rigidity of the special-player stabilizer game to
the full rigidity property of the stabilizer game itself. Consistency
and another state dependent distance measure of two quantum
measurements appeared before in the analysis of non-local games and
are intensively used in many proofs of this paper.

We then consider the multi-qubit stabilizer game, the non-local game
version of the stabilizer encoding check in~\cite{FV15}. We prove a
partial rigidity theorem for the multi-qubit stabilizer game. This is
achieved by establishing the approximate commutativity of reflections
corresponding to measurement specifications on different qubits, the
proof of which uses the consistency property again in an essential
way. Finally, the multi-qubit stabilizer game is performed in the game
for the local Hamiltonian problem with high probability to regularize
the behavior of the players.

Two additional difficulties arise as the verifier only has limited
access to the quantum witness state. First, since the verifier can
only enforce Pauli measurements on the witness state, the measurement
of energy of the local Hamiltonian uses Pauli measurements only. This
will be less efficient than having the quantum data and measuring
directly the POVM corresponding to each term of the Hamiltonian as is
done in~\cite{FV15}, but there will only be a constant overhead
compared to the direct measurement approach.

Second, for convenience, we work with stabilizer codes whose
generators are the tensor products of Pauli $I$, $X$, $Z$ operators.
Alternatively, one may design an extended version of the stabilizer
game using ideas from the extended CHSH game proposed in~\cite{RUV13}
to include Pauli $Y$ operators in the problem. The current form of the
stabilizer game is, however, much easier to analyze. As a result, by
the rigidity property, the verifier only have access to Pauli $X$, $Z$
measurements on the quantum witness state. This requires that the
Hamiltonian in the local Hamiltonian problem has terms in the real
linear span of tensor products of $I$, $X$ and $Z$ operators.
Fortunately, as observed in~\cite{BL08}, such restricted version of
the local Hamiltonian problem remains \class{QMA}-complete.

\subsection{Open Problems}

We briefly mention several related open problems. One major weakness
of our result for the multi-prover interactive proofs with entangled
provers in Corollary~\ref{cor:MIP*} is the exponential small gap
between the completeness and soundness parameters. It is an intriguing
and challenging problem to improve this and show that
$\class{QMA}_{\class{EXP}}$ is contained in $\class{MIP}^*$.

Second, as we don't have any upper bound on the power of non-local
games, it remains possible to prove even stronger hardness results than
the \class{QMA}-hardness shown in this paper. A possible candidate
problem is to show \class{QMA}(2)-hardness of non-local games.

Finally, it is an interesting question to ask whether one can improve
the completeness parameter of our protocol. For example, are there
non-local games with perfect completeness for $\class{QMA}_1$? It
seems that, to achieve this, we need to redesign the stabilizer game
and simplify the form of the local Hamiltonians to improve the
efficiency of the energy measurement.

\subsection{Organizations}

The rest of the paper is organized as follows. In Sec.~\ref{sec:prel},
we introduce the notions and related concepts used in this paper. The
stabilizer games are introduced and analyzed in Sec.~\ref{sec:stab}.
The non-local game for the local Hamiltonian problem is given and
analyzed in Sec.~\ref{sec:lh}.

\section{Preliminaries}

\label{sec:prel}

\subsection{Notions}

We use calligraphic $\H$ to denote Hilbert spaces, and $\Density(\H)$,
$\Lin(\H)$, $\Herm(\H)$, $\Pos(\H)$ to denote the set of density
operators, bounded linear operators, Hermitian operators and positive
semidefinite operators on $\H$. Two-dimensional Hilbert spaces
$\complex^2$ corresponding to a qubit is denoted by $\B$. For two
Hermitian operators $M,N\in \Herm(\H)$, we write $M \le N$ to mean
$N-M\in \Pos(\H)$. For matrix $M$, $\abs{M}$ is defined to be
$\sqrt{M^\dagger M}$. For a string $x$, $\abs{x}$ denotes its length.
For a positive integer $k$, $[k]$ is the abbreviation of the set
$\{1,2,\ldots, k\}$. For two complex numbers $x$, $y$, we use
$x \approx_\epsilon y$ as a shorthand notion for
$\abs{x-y}\le O(\epsilon)$. Two-qubit unitary gates $\CNOT$ and
$\SWAP$ are defined as
\begin{equation*}
  \CNOT \ket{i,j} = \ket{i}\ket{i\oplus j},\quad
  \SWAP \ket{i,j} = \ket{j,i},
\end{equation*}
for $i,j\in\{0,1\}$.

Let $\rho\in \Density(\H)$ be a quantum state on $\H$. For operators
$M,N\in \Lin(\H)$, introduce the following notions
\begin{subequations}
  \begin{align}
    \tr_\rho (M) & = \tr(M\rho),\\
    \ip{M}{N}_\rho & = \tr_\rho (M^\dagger N),\\
    \norm{M}_\rho & = \sqrt{\ip{M}{M}_\rho}.
  \end{align}
\end{subequations}
It is straightforward to verify that $\ip{\cdot}{\cdot}_\rho$ is a
semi-inner-product, $\norm{\cdot}_\rho$ is a seminorm and they become
an inner product and a norm, respectively, when $\rho$ is a full-rank
state. In particular, the Cauchy-Schwarz inequality holds
\begin{equation*}
  \abs{\ip{M}{N}_\rho} \le \norm{M}_\rho \norm{N}_\rho,
\end{equation*}
or more explicitly,
\begin{equation*}
  \abs{\tr_\rho (M^\dagger N)} \le \left[ \tr_\rho(M^\dagger M) \,
    \tr_\rho(N^\dagger N) \right]^{1/2}.
\end{equation*}

For state $\rho\in \Density(\H_A \otimes \H_B)$, and an operator
$M\in \Lin(\H_A)$, we may also write $\tr_\rho(M)$ even though the
state $\rho$ and the operator $M$ do not act on the same space. In
this case, it is understood that $\tr_\rho(M) = \tr_{\rho_A} (M)$
where $\rho_A$ is the reduced state of $\rho$ on system $A$. This is
one reason that makes $\tr_\rho(\cdot)$ easy to use as it is not
necessary to specify the correct reduced state explicitly all the
time.

A quantum measurement is described by a collection $M=\{M^a\}$ of
measurement operators. These are operators acting on the state space
$\H$ of the system being measured. In this paper, we always use
superscript to index the measurement outcomes and subscript to index
different quantum measurements. The measurement operators satisfy the
completeness equation
\begin{equation*}
  \sum_a (M^a)^\dagger M^a = I.
\end{equation*}
If the state of the system before the measurement is
$\rho\in \Density(\H)$, the probability that outcome $a$ occurs is
given by
\begin{equation*}
  \tr \bigl( (M^a)^\dagger M^a\rho \bigr) = \norm{M^a}_\rho^2,
\end{equation*}
and the post-measurement state of the system is
\begin{equation*}
  \frac{M^a \rho (M^a)^\dagger}{\norm{M^a}_\rho^2}.
\end{equation*}

A positive operator valued measure (POVM) is a collection of positive
semidefinite operators $\{M^a\}$, such that
\begin{equation*}
  \sum_a M^a = I.
\end{equation*}
POVMs give descriptions of quantum measurements when the
post-measurement state is not important in the analysis. The
probability that the measurement has outcome $a$ on state $\rho$ is
given by $\tr_\rho (M^a)$.

A projective quantum measurement is described by $\{M^a\}$ where each
operator $M^a$ is a projection and $\sum_a M^a = I$. A reflection $R$
is a Hermitian matrix squared to $I$. It naturally describes a
two-outcome projective quantum measurement $\{R^a\}$ where
\begin{equation*}
  R^a = \frac{I+(-1)^a R}{2},
\end{equation*}
for $a=0,1$. The Pauli operators
\begin{equation*}
  I =
  \begin{bmatrix}
    1 & 0\\
    0 & 1
  \end{bmatrix},\quad X =
  \begin{bmatrix}
    0 & 1\\
    1 & 0
  \end{bmatrix},\quad Y =
  \begin{bmatrix}
    0 & -i\\
    i & \phantom{-}0
  \end{bmatrix},\quad Z =
  \begin{bmatrix}
    1 & \phantom{-}0\\
    0 & -1
  \end{bmatrix},
\end{equation*}
are examples of $2$-by-$2$ reflections. We also use
$\sigma_0, \sigma_1, \sigma_2, \sigma_3$ to represent the four Pauli
operators. A multi-qubit Pauli operator is of $XZ$-form if each tensor
factor is one of $I$, $X$ and $Z$. A Hermitian matrix $H$ is of
$XZ$-form if it is in the real linear span of $XZ$-from Pauli
operators.

Define $X'$ and $Z'$ as
\begin{equation}
  \label{eq:XZ'}
  X' = \frac{X+Z}{\sqrt{2}},\qquad Z' = \frac{X-Z}{\sqrt{2}},
\end{equation}
and $W$ as
\begin{equation}
  \label{eq:W}
  W = \cos (\pi/8) X + \sin(\pi/8) Z.
\end{equation}
It is easy to verify that $W$ is a reflection and
\begin{alignat*}{2}
  X & = W X' W, & \qquad X' & = W X W,\\
  Z & = W Z' W, & \qquad Z' & = W Z W.
\end{alignat*}
That is, under the conjugation of $W$, $X'$ and $Z'$ are mapped to $X$
and $Z$ respectively, and vice versa. The reflections $X$, $Z$, $X'$,
$Z'$ and $W$ are illustrate in Fig.~\ref{fig:refs}. They play an
important role in the CHSH game and the stabilizer games introduced in
this paper.

\begin{figure}[ht]
  \centering
  \includegraphics{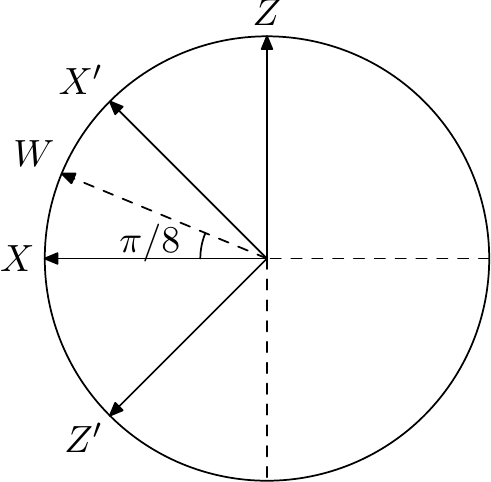}
  \caption{Reflections $X$, $Z$, $X'$, $Z'$, $W$ in the $x,z$-plane of
    the Bloch sphere.}
  \label{fig:refs}
\end{figure}

\subsection{Non-Local Games}

A multi-player one-round game involves two or more players and a
verifier who communicates with the players classically in one round.
The verifier samples questions and sends them out to the players and
expects to receive answers back. He then decides whether to accept or
reject based on the questions and answers. The players are allowed to
agree on a strategy before the game starts, but cannot communicate
with each other during the game.

Let there be $r$ players, $(1), (2), \ldots, (r)$. Let $\Gamma^{(i)}$
be a finite set of questions for player $(i)$ and $\Lambda^{(i)}$ be a
finite set of possible answers from player $(i)$. An $r$-player game
is defined by a distribution $\pi$ over $\prod_{i=1}^r \Gamma^{(i)}$
and a function
$V: \prod_{i=1}^r \Lambda^{(i)} \times \prod_{i=1}^r \Gamma^{(i)}
\rightarrow [0,1]$,
specifying the acceptance probability. By a convexity argument, it
suffices to consider the strategy of classical players described by
functions $f^{(i)} : \Gamma^{(i)} \rightarrow \Lambda^{(i)}$. The
value of the strategy is the acceptance probability
\begin{equation*}
  \omega = \E_{q\sim \pi} V (a(q),q),
\end{equation*}
for $q=(q_1,q_2,\ldots,q_r)$ distributed according to $\pi$ and
$a(q) = \bigl( f^{(1)}(q_1), f^{(2)}(q_2), \ldots, f^{(r)}(q_r)
\bigr)$.
The classical value of the game is the maximum of the values of all
classical strategies. An XOR game is a multi-player game in which each
player answers a bit $a_i$ and the verifier accepts or rejects
depending only on the parity of $\bigoplus_{i=1}^r a_i$. More
precisely, there is a function
$\hat{V}:\{0,1\} \times \prod_{i=1}^r \Gamma^{(i)} \rightarrow [0,1]$
such that
\begin{equation*}
  V(a,q) = \hat{V} \bigl( \bigoplus_{i=1}^r a_i,q \bigr),
\end{equation*}
for all $q\in \prod_{i=1}^r\Gamma^{(i)}$,
$a\in \prod_{i=1}^r\Lambda^{(i)}$. In most of the games considered in
this paper, the verifier accepts or rejects only depending on the
parity of some, not all, of the answer bits. We call them generalized
XOR games.

In a non-local game, the players are allowed to share an arbitrary
entangled state before the game starts. A quantum strategy $\S$ for
the non-local game is described by the shared state $\rho$, the
measurements $\bigl\{ M^{(i)}_{q_i} \bigr\}$ that player $(i)$
performs when the question is $q_i\in \Gamma^{(i)}$. The value of the
strategy is defined as
\begin{equation*}
  \omega^*(\S) = \E_{q\sim \pi} \sum_a \biggl[ \tr_\rho
  \bigl(\bigotimes_{i=1}^r M^{(i),a_i}_{q_i} \bigr) V(a,q) \biggr],
\end{equation*}
for $a=(a_1,a_2,\ldots,a_r)$ and $q=(q_1,q_2,\ldots,q_r)$. The
non-local value of the game is the supremum of the values of all
quantum strategies.

The CHSH game~\cite{CHSH69} is arguably one of the most important
non-local games. It arises from the study of fundamental questions in
quantum mechanics such as entanglement and non-locality via Bell
inequalities~\cite{Bel64}. The CHSH game is a two-player XOR game. The
verifier samples two bits $q_1$ and $q_2$ independently and uniformly
at random, sends $q_1$ to the first player and $q_2$ to the second
player. The verifier accepts if and only if two answer bits $a_1$ and
$a_2$ satisfy $a_1\oplus a_2 = q_1 \land q_2$. The classical value of
the game is $3/4=0.75$ and the non-local value of the game is
$\omega^*_{\text{CHSH}} = (2+\sqrt{2})/4 \approx 0.85$~\cite{Tsi80}.

In an optimal strategy for the CHSH game, the players share an EPR
state $(\ket{00}+\ket{11})/\sqrt{2}$, and the first player obtains the
answer by measuring $X$ (or $Z$) if the question is $0$ (or $1$
respectively), while the second player measures $X'$ (or $Z'$) in
Eq.~\eqref{eq:XZ'} if the question is $0$ (or $1$). The rigidity
property of the CHSH game roughly states that this is essentially the
only strategy for the players to achieve the non-local value, up to
local isometries. Furthermore, any strategy that has value close to
the game value must be close to this optimal strategy in some sense.
Rigidity of the CHSH game and other non-local games has found
interesting applications in certifiable randomness
generation~(e.g.,~\cite{Col06,PAM+10,VV12,MS14}), device-independent
quantum cryptography~(e.g.,~\cite{ABG+07,VV14,MS14}) and classical
command of quantum systems~\cite{RUV13}.

\subsection{Quantum Proofs and Local Hamiltonian Problems}

The idea of efficient proof verification of \class{NP} has a natural
quantum generalization given in the following definition.
\begin{definition}
  The complexity class \class{QMA} is the set of promise problems
  $L=(L_{\text{yes}},L_{\text{no}})$ such that there is a polynomial
  $p(\cdot)$ and a quantum polynomial-time verifier $V$, and
  \begin{itemize}
  \item {\it Completeness.\/} If $x\in L_{\text{yes}}$, there exists a
    state $\ket{\psi}$ of $p(\abs{x})$ qubits,
    \begin{equation*}
      \Pr \bigl[ V \text{ accepts } \ket{x}\otimes \ket{\psi} \bigr]
      \ge \frac{2}{3},
    \end{equation*}
  \item {\it Soundness.\/} If $x\in L_{\text{no}}$, then for all state
    $\ket{\psi}$ of $p(\abs{x})$ qubits,
    \begin{equation*}
      \Pr \bigl[ V \text{ accepts } \ket{x}\otimes \ket{\psi} \bigr]
      \le \frac{1}{3}.
    \end{equation*}
  \end{itemize}
\end{definition}

\begin{definition}[Local Hamiltonian Problems]
  An instance of the $k$-local Hamiltonian problem of $n$-qubits is
  described by the tuple $(H,a,b)$, where the Hamiltonian
  $H=\sum_{j=1}^m H_j$ and each term $H_j$ acts non-trivially on at
  most $k$-qubits, $H_j$ is positive semidefinite and
  $\norm{H_j} \le 1$, $a,b\in \real$ are numbers satisfying
  $b-a\ge 1/\poly(n)$. Let the ground state energy of the Hamiltonian
  $H$ be $\lambda_{\min} = \min_{\rho \in \Density} \ip{H}{\rho}$. In
  the $k$-local Hamiltonian problem, $(H,a,b)$ is a yes-instance if
  $\lambda_{\min} \le am$ and a no-instance if
  $\lambda_{\min} \ge bm$.
\end{definition}

The quantum analog of the Cook-Levin theorem by Kitaev states that the
$k$-local Hamiltonian problem is \class{QMA}-complete for
$k\ge 5$~\cite{Kit99,AN02}. This was later improved to $2$-local and
more physical Hamiltonians in a series of works~(see
e.g.,~\cite{KKR06,OT08,CM14}).

\subsection{Quantum Error Correction and Stabilizer Formalism}

A quantum error correcting code encodes a number of qubits, called the
logical qubits, into a larger number of physical qubits with the aim
of protecting the quantum information in the logical qubits from
certain types of noises.

The stabilizer formalism provides a convenient language and great
amount of examples of quantum error correcting codes. We present
several relevant definitions and facts about the stabilizer codes and
refer the reader to the thesis of Gottesman~\cite{Got97} for more
details. Let $\P_r$ be the group generated by the $r$-fold tensor
product of Pauli operators
\begin{equation*}
  \P_r = \biggl\{ e^{i\phi} \bigotimes_{j=1}^r P_j, \text{ for }\phi \in
  \{0,\pi/2,\pi,3\pi/2\},\; P_j \in \{I,X,Y,Z\} \biggr\}.
\end{equation*}
A stabilizer $\Stab$ is an abelian subgroup of $\P_r$ not containing
$-I^{\otimes r}$. The stabilizer provides a succinct description of a
subspace of $(\complex^2)^{\otimes r}$, the simultaneous
$+1$-eigenspace of the operators in $\Stab$. This subspace is called
the code space of the stabilizer. A set of operators in $\Stab$ is
called the generators of $\Stab$ if they generate the group $\Stab$.

The weight of an operator in $\P_r$ is the number of non-identity
tensor factors in it. Let $C(\Stab)$ be the centralizer of $\Stab$ in
$\P_r$, the set of operators in $\P_r$ that commutes with $\Stab$. The
distance of the stabilizer code is $d$ if there is no operator of
weight less than $d$ in $C(\Stab)-\Stab$. The logical $X$ and $Z$
operators $L_X$ and $L_Z$ are a pair of anti-commuting operators in
$C(\Stab)-\Stab$.

As a simple example, the operators $XX$ and $ZZ$ generate a stabilizer
for the EPR state. The famous five-qubit code~\cite{LMPZ96,BDSW96} has
a stabilizer representation given in Fig.~\ref{fig:5code1}. The
five-qubit code encodes one logical qubit in five physical qubits and
has distance $3$. The logical $X$ and $Z$ operators are
$X^{\otimes 5}$ and $Z^{\otimes 5}$. This will be the stabilizer code
we use most of the time as the example. Operators $X^{\otimes 4}$ and
$Z^{\otimes 4}$ generate the stabilizer for the four-qubit quantum
error detecting code. It has distance $2$ and encodes two qubits. The
operators $XXII$ and $ZIZI$ form a pair of anti-commuting operators
and can serve as the logical $X$ and $Z$ operators. There is another
pair of the logical operators for the other encoded qubit that we do
not use.

\subsection{State-Dependent Distance Measures of Quantum Measurements}

\label{sec:dist}

We introduce a distance measure and a consistency measure of quantum
measurements that will be helpful in our treatment of non-local games.
They grow out of the study of non-local games~\cite{IKM09,IV12,Vid13}
and may be useful in more general contexts. A common feature of them
is that they are both state-dependent.

A general question that one usually faces is what the consequences are
if measurement $\{M_1^a\}$ is used in place of $\{M_0^a\}$ in a
non-local game by one of the players. Will it change the overall
analysis at lot? More concretely, suppose that the state of a joint
quantum system $\H_A$ and $\H_B$ is $\rho$ and that $\{M_0^a\}$,
$\{M_1^a\}$ are quantum measurements on system $A$. The
post-measurement states are
\begin{equation}
  \label{eq:postms}
  \rho_i = \sum_a \op{a}{a} \otimes M_i^a \rho (M_i^a)^\dagger,
\end{equation}
for $i = 0,1$, respectively, depending on which measurement is
performed. By the monotonicity of the trace distance, the difference
will be bounded by $\dtr(\rho_0, \rho_1)$ no matter what operation
follows the measurement. In particular, in a non-local game, if Bob
measures on his system $\H_B$ and then the verifier makes the
decision, the acceptance probabilities will differ by at most
$\dtr(\rho_0, \rho_1)$.

The quantity defined next provides a way to bound the distance
$\dtr(\rho_0, \rho_1)$.

\begin{definition}
  For two quantum measurements $M_i = \bigl\{ M_i^a \bigr\}$ with
  $i=0,1$ that have the same set of possible outcomes, define
  \begin{equation}
    \label{eq:drho}
    d_\rho(M_0,M_1) = \Bigl[ \sum_a \norm{M_0^a-M_1^a}_\rho^2
    \Bigr]^{1/2}.
  \end{equation}
  More explicitly,
  \begin{equation}
    \label{eq:drho2}
    d_\rho(M_0,M_1) = \biggl[2-2\Re \sum_a \tr_\rho \bigl(
    (M_0^a)^\dagger M_1^a \bigr) \biggr]^{1/2}.
  \end{equation}
\end{definition}

\begin{lemma}
  \label{lem:drho}
  Let $M_i = \bigl\{ M_i^a \bigr\}$ for $i=0,1$ be two quantum
  measurements with the same set of possible outcomes, and $\rho_i$ be
  the post-measurement states in Eq.~\eqref{eq:postms}. Then
  \begin{equation*}
    \dtr(\rho_0, \rho_1) \le d_\rho (M_0, M_1).
  \end{equation*}
\end{lemma}

\begin{proof}
  Let $\ket{\psi} \in \H_A \otimes \H_B \otimes \H_C$ be a
  purification of $\rho$. Then
  \begin{equation*}
    \begin{split}
      \dtr(\rho_0,\rho_1) & \le \dtr \bigl(\sum_a \ket{a} \otimes
      (M_0^a \ket{\psi}), \sum_a \ket{a} \otimes (M_1^a \ket{\psi})
      \bigr)\\
      & = \biggl[ 1- \Bigl| \sum_a \bra{\psi} (M_0^a)^\dagger M_1^a
      \ket{\psi} \Bigr|^2 \biggr]^{1/2}\\
      & \le \biggl[ 2 \biggl( 1- \Bigl| \sum_a \bra{\psi}
      (M_0^a)^\dagger
      M_1^a \ket{\psi} \Bigr| \biggr) \biggr]^{1/2}\\
      & \le \biggl[2-2\Re \sum_a \tr_\rho \bigl( (M_0^a)^\dagger M_1^a
      \bigr) \biggr]^{1/2}\\
      & = d_\rho(M_0,M_1).
    \end{split}
  \end{equation*}
  The first inequality follows from the monotonicity of the trace
  distance. The second line follows by a direct calculation of the
  trace distance for two pure states. The third line is from
  $1-x^2\le 2(1-x)$ for $x\in[0,1]$.
\end{proof}

As discussed above, a direct corollary of the above lemma is that when
measurement $M_0$ is replaced with $M_1$ in a strategy for a non-local
game using shared state $\rho$, the acceptance probability change by
at most $d_\rho (M_0, M_1)$. This claim works for all types of quantum
measurements including the general quantum measurement, POVMs,
projective quantum measurements and binary projective measurements
described by reflections.

For $M_i = \bigl\{ M_i^a \bigr\}$, $i=0,1$, describing two POVMs that
satisfy $\sum_a M_i^a = I$, define the corresponding distance as
\begin{equation*}
  d_\rho(M_0,M_1) = \biggl[2-2\Re \sum_a \tr_\rho \bigl( \sqrt{M_0^a}
  \sqrt{M_1^a} \bigr) \biggr]^{1/2}.
\end{equation*}
For projective measurements $M_i = \bigl\{ M_i^a \bigr\}$, define
\begin{equation*}
  d_\rho(M_0,M_1) = \biggl[2-2\Re \sum_a \tr_\rho \bigl( M_0^a M_1^a
  \bigr) \biggr]^{1/2}.
\end{equation*}
Finally, for reflections $R_0, R_1$, let
\begin{equation*}
  R_i^a = \frac{I+(-1)^a R_i}{2}
\end{equation*}
be the projective measurement operators correspond to $R_i$. Define
\begin{equation*}
  \begin{split}
    d_\rho(R_0,R_1) & = d_\rho(\{R_0^a\}, \{R_1^a\})\\
    & = \biggl[2 - 2\Re \sum_a \tr_\rho \bigl( \frac{I+(-1)^a
      R_0}{2} \frac{I+(-1)^a R_1}{2} \bigr) \biggr]^{1/2}\\
    & = \bigl[1 - \Re \tr_\rho \bigl( R_0 R_1 \bigr) \bigr]^{1/2}.
  \end{split}
\end{equation*}

It is easy to verify that $d_\rho$ satisfy the triangle inequality.

\begin{lemma}
  Let $M_0$, $M_1$, $M_2$ be three measurements on state $\rho$. Then
  \begin{equation*}
    d_\rho(M_0,M_2) \le d_\rho(M_0,M_1) + d_\rho(M_1,M_2).
  \end{equation*}
\end{lemma}

The next important quantity measures the consistency of two quantum
measurements.

\begin{definition}
  Let $\rho \in \Density(\H_A \otimes \H_B)$ be the shared state
  between system $A$ and $B$, let $M= \bigl\{ M^a \bigr\}$,
  $N= \bigl\{ N^a \bigr\}$ be POVMs on system $A$ and $B$ respectively
  having the same set of possible outcomes. Define the consistency of
  $M$, $N$ on state $\rho$ as
  \begin{equation}
    \label{eq:cons}
    C_\rho(M,N) = \sum_a \tr_\rho (M^a\otimes N^a).
  \end{equation}
  $M$ and $N$ are called $\epsilon$-consistent on state $\rho$ if
  $C_\rho(M,N) \ge 1 - \epsilon$.
\end{definition}

For two reflections $R$, $S$, let $\{R^a\}$, $\{S^a\}$ be their
corresponding projective measurements. Define
\begin{equation}
  \label{eq:consref}
  C_\rho(R,S) = C_\rho(\{R^a\}, \{S^a\}) = \frac{1+\tr_\rho(R\otimes
    S)}{2}.
\end{equation}
The condition $\tr_\rho(R\otimes S) \approx_\epsilon 1$, or
equivalently, $R,S$ are $O(\epsilon)$-consistent on $\rho$, can be
thought of as a quantitative way of saying that $\rho$ is
approximately stabilized by $R\otimes S$.

The consistency of measurements puts strong structural constraints on
the strategies of non-local game. It will be a key ingredient in our
proof of the main result. In the following lemma, it is proved that if
two measurements $M_0$, $M_1$ are consistent with the same measurement
$N$, then $M_0$, $M_1$ must be close to each other in terms of the
distance $d_\rho$.

\begin{lemma}
  \label{lem:cons}
  Let $\rho \in \Density(\H_A \otimes \H_B)$ be a state on system $A$
  and $B$. Let $M_i= \bigl\{ M_i^a \bigr\}$ for $i=0,1$ be POVMs on
  system $A$, $N= \bigl\{ N^a \bigr\}$ be a POVM on system $B$. If
  \begin{equation*}
    C_\rho(M_i,N) \ge 1 - \epsilon,
  \end{equation*}
  for $i=0,1$, then
  \begin{equation*}
    d_\rho(M_0, M_1) \le O(\sqrt{\epsilon}).
  \end{equation*}
\end{lemma}

\begin{proof}
  First prove that
  \begin{equation}
    \label{eq:conslemma1}
    \sum_a \tr_\rho \biggl[ \bigl( 1-\sqrt{M_0^a} \bigr) \bigl(
    1-\sqrt{M_1^a} \bigr)\otimes N^a \biggr]
    \approx_\epsilon 0.
  \end{equation}
  By the Cauchy-Schwarz inequality, the absolute value of the left
  hand side is at most
  \begin{equation*}
    \begin{split}
      & \sum_a \sqrt{\tr_\rho \biggl[ \bigl( 1-\sqrt{M_0^a}
        \bigr)^2\otimes N^a \biggr] \tr_\rho \biggl[ \bigl(
        1-\sqrt{M_1^a} \bigr)^2\otimes N^a \biggr] }\\
      \le & \sum_a \sqrt{\tr_\rho \bigl[ \bigl( 1-M_0^a \bigr) \otimes
        N^a \bigr] \tr_\rho \bigl[ \bigl( 1-M_1^a \bigr) \otimes N^a
        \bigr] }\\
      \le & \sqrt{\sum_a \tr_\rho \bigl[ \bigl( 1-M_0^a \bigr) \otimes
        N^a \bigr] \sum_a \tr_\rho \bigl[ \bigl( 1-M_1^a \bigr)
        \otimes N^a \bigr] }\\
      \le & \epsilon,
    \end{split}
  \end{equation*}
  where the first inequality follows from the fact that
  $(1-\sqrt{x})^2 \le 1-x$ for $x\in[0,1]$, the second inequality is
  another Cauchy-Schwarz inequality and the last inequality follows
  from the condition that $C_\rho(M_i,N)\ge 1-\epsilon$.

  Similarly, by using the Cauchy-Schwarz inequality twice, we have
  \begin{equation}
    \label{eq:conslemma2}
    \sum_a \tr_\rho \biggl[ \sqrt{M_0^a} \sqrt{M_1^a}  \otimes \bigl(
    1 - N^a \bigr) \biggr] \approx_\epsilon 0.
  \end{equation}

  Adding Eqs.~\eqref{eq:conslemma1} and~\eqref{eq:conslemma2} gives
  \begin{equation*}
    \begin{split}
      \sum_a \tr_\rho \biggl[ \sqrt{M_0^a}\sqrt{M_1^a} \biggr] &
      \approx_\epsilon \sum_a \tr_\rho \biggl[\sqrt{M_0^a} \otimes N^a
      \biggr] + \sum_a \tr_\rho \biggl[\sqrt{M_1^a} \otimes N^a
      \biggr] - 1\\
      & \approx_\epsilon 1+1-1=1.
    \end{split}
  \end{equation*}
  The lemma follows by the definition of $d_\rho$ for POVMs.
\end{proof}

In the analysis, it is useful to have a quantity characterizing the
approximate commutativity of two projective measurements $M=\{ M^a \}$
and $N=\{ N^a \}$ on a state $\rho$. The quantity we choose is
\begin{equation*}
  \sum_a \norm{[M^a,N^a]}_\rho^2.
\end{equation*}

For reflections $R,S$, let $\{ R^a \}$ and $\{ S^a \}$ be the
projective measurements correspond to $R$ and $S$ respectively. The
commutativity of these two measurements on state $\rho$
\begin{equation*}
  \begin{split}
    \sum_{a\in \{0,1\}} \norm{[R^a,S^a]}^2_\rho & = \sum_{a\in
      \{0,1\}} \norm{\left[ \frac{I+(-1)^a R}{2},\frac{I+(-1)^a S}{2}
      \right]}_\rho^2\\
    & = \frac{1}{8} \norm{[R,S]}_\rho^2.
  \end{split}
\end{equation*}
For this reason, $\norm{[R,S]}_\rho^2$ equivalently serves as a bound
on the approximate commutativity of the two projective measurements
defined by $R,S$.

\section{Stabilizer Games and Rigidity}

\label{sec:stab}

\subsection{CHSH Game Revisited}

Before introducing the stabilizer games, it is beneficial to revisit
the CHSH game in the stabilizer formalism.

Recall that the EPR state
$\ket{\Phi} = (\ket{00} + \ket{11}) / \sqrt{2}$ is a stabilizer state
defined by two generators $g_1=XX$ and $g_2=ZZ$, and the eigenstate of
eigenvalue $2$ of the operator $g_1 + g_2$. If a verifier has trusted
measuring devices, it suffices to perform the projective measurements
associated with the reflections $g_1, g_2$ to check whether the state
is an EPR state. However, this simple measurement setting does not
correspond to any non-trivial schemes that allow device-independent
certification of the EPR state.

In the CHSH game, one of the two players is asked to measure her share
of $\ket{\Phi}$ with $X$, $Z$, the other is asked to measure $X'$,
$Z'$ in Eq.~\eqref{eq:XZ'}, the $\pi/4$ rotated versions of $X$, $Z$.
This motivates us to consider the generators $g_1' = XX'$ and
$g_2' = ZZ'$. By the conjugation relation of $X$, $Z$ and $X'$, $Z'$,
they generate a stabilizer for the state
$\ket{\Phi'} = I\otimes W \ket{\Phi}$ and
\begin{equation}
  \label{eq:CHSH'}
  \bigbra{\Phi'} \sum_{i=1}^2g_i' \bigket{\Phi'} = 2.
\end{equation}
Expanding the $X'$, $Z'$ in $g_1'$ and $g_2'$ using $X$ and $Z$ gives
four operators
\begin{equation}
  \label{eq:CHSH4}
  h_1 = XX,\quad h_2 = XZ,\quad h_3 = ZX,\quad h_4 = -ZZ,
\end{equation}
such that $h_1 + h_2 = \sqrt{2} g_1'$ and $h_3 + h_4 = \sqrt{2} g_2'$.
The four operators $h_i$ in Eq.~\eqref{eq:CHSH4} recover exactly the
CHSH game by encoding the Pauli operators $X$, $Z$ in $h_i$ as
questions $0,1$ and the sign $\pm 1$ of $h_i$ as the expected parity
of answers. The Eq.~\eqref{eq:CHSH'} becomes
\begin{equation*}
  \bigbra{\Phi'} \sum_{i=1}^4 h_i \bigket{\Phi'} = 2\sqrt{2},
\end{equation*}
which is an explanation of the $\sqrt{2}$ quantum advantage in the
CHSH game.

\subsection{Special-Player Stabilizer Game}

In this section, we introduce the stabilizer games with a special
player. The construction works with any non-trivial stabilizer code
that has a set of generators all in the tensor product form of
$I,X,Z$.

Consider the generators of the stabilizer group for the five-qubit quantum
code in Fig.~\ref{fig:5code1}. Its code space is the two-dimensional
eigenspace of eigenvalue $4$ of the operator $\sum_{j=1}^4g_j$, where
$g_j$'s are operators in Fig.~\ref{fig:5code1}.

\begin{figure}[ht]
  \begin{subfigure}[t]{.45\linewidth}
    \centering
    \begin{tabular}[c]{c|c@{}c@{}c@{}c@{}c}
      \hline\hline
      Name & \multicolumn{5}{c}{Operator}\\
      \hline
      $g_1$ & $I$ & $X$ & $Z$ & $Z$ & $X$\\
      $g_2$ & $X$ & $I$ & $X$ & $Z$ & $Z$\\
      $g_3$ & $Z$ & $X$ & $I$ & $X$ & $Z$\\
      $g_4$ & $Z$ & $Z$ & $X$ & $I$ & $X$\\
      \hline\hline
    \end{tabular}
    \caption{Standard generators}
    \label{fig:5code1}
  \end{subfigure}
  \begin{subfigure}[t]{.45\linewidth}
    \centering
    \begin{tabular}[c]{c|c@{}c@{}c@{}c@{}c}
      \hline\hline
      Name & \multicolumn{5}{c}{Operator}\\
      \hline
      $g'_1$ & $I$ & $X$ & $Z$ & $Z$ & $X'$\\
      $g'_2$ & $X$ & $I$ & $X$ & $Z$ & $Z'$\\
      $g'_3$ & $Z$ & $X$ & $I$ & $X$ & $Z'$\\
      $g'_4$ & $Z$ & $Z$ & $X$ & $I$ & $X'$\\
      \hline\hline
    \end{tabular}
    \caption{Generators with the last qubit rotated.}
    \label{fig:5code2}
  \end{subfigure}
  \caption{Stabilizer generators for the five-qubit code.}
\end{figure}

Motivated by the CHSH game, we apply the $\pi/4$-trick to the last
column of the four generators in Fig.~\ref{fig:5code1}. That is, we
replace $X$ and $Z$ in the last column with $X'$ and $Z'$ in
Eq.~\eqref{eq:XZ'}. This gives us another set of generators, as in
Fig.~\ref{fig:5code2}, which generates the stabilizer for the
five-qubit code with the last qubit rotated by the single-qubit
unitary $W$ in Eq.~\eqref{eq:W}. If $\ket{\psi}$ is in the code space
of the five-qubit code, state
\begin{equation}
  \label{eq:psi'}
  \ket{\psi'}=(I\otimes W)\ket{\psi}
\end{equation}
is in the eigenspace of $\sum_{j=1}^4g'_j$ with eigenvalue $4$.

Expanding the primed $X,Z$ operators in each of the generators in
Fig.~\ref{fig:5code2} into $X,Z$, we get a set of eight operators
$h_j$ as in Fig.~\ref{fig:5code3}. For state $\ket{\psi'}$ in
Eq.~\eqref{eq:psi'},
\begin{equation}
  \label{eq:h_i}
  \bigbra{\psi'} \sum_{j=1}^8 h_j \bigket{\psi'} = 4\sqrt{2}.
\end{equation}

The table in Fig.~\ref{fig:5code4} is obtained by translating the
operators $I$, $X$, $Z$ in Fig.~\ref{fig:5code3} to the questions in
the alphabet of $*,0,1,2,3$. The $I$ operator is always translated to
$*$, which denotes a null question. The verifier will not ask anything
to the player and do not expect any answers if the question is the
null question $*$. For convenience, we sometimes assume that the
verifier will replace $*$ with either $0$ or $1$ as the question and
ignore the answers corresponding to this question. The operators $X$,
$Z$ are translated to $0$, $1$ respectively in the first four columns
and to $2$, $3$ respectively in the last column. One can of course
also use $0$, $1$ in the last the column and the change is only for
later convenience. Finally, the parity column is read off from the
$\pm 1$ signs in the $h_i$ operators.

\begin{figure}[ht]
  \begin{subfigure}[t]{.45\linewidth}
    \centering
    \begin{tabular}[c]{c|c@{}c@{}c@{}c@{}c@{}c}
      \hline\hline
      Name & \multicolumn{6}{c}{Operator}\\
      \hline
      $h_1$ & & $I$ & $X$ & $Z$ & $Z$ & $X$\\
      $h_2$ & & $I$ & $X$ & $Z$ & $Z$ & $Z$\\
      \hline
      $h_3$ & & $X$ & $I$ & $X$ & $Z$ & $X$\\
      $h_4$ & $-$ & $X$ & $I$ & $X$ & $Z$ & $Z$\\
      \hline
      $h_5$ & & $Z$ & $X$ & $I$ & $X$ & $X$\\
      $h_6$ & $-$ & $Z$ & $X$ & $I$ & $X$ & $Z$\\
      \hline
      $h_7$ & & $Z$ & $Z$ & $X$ & $I$ & $X$\\
      $h_8$ & & $Z$ & $Z$ & $X$ & $I$ & $Z$\\
      \hline\hline
    \end{tabular}
    \caption{Eight operators obtained from Fig.~\ref{fig:5code2}.}
    \label{fig:5code3}
  \end{subfigure}
  \begin{subfigure}[t]{.45\linewidth}
    \centering
    \begin{tabular}[c]{c|c@{\hskip 1ex}c@{\hskip 1ex}c@{\hskip
          1ex}c@{\hskip 1ex}c}
      \hline\hline
      Parity & \multicolumn{5}{c}{Question}\\
      \hline
      $0$ & $*$ & $0$ & $1$ & $1$ & $2$\\
      \hline
      $0$ & $*$ & $0$ & $1$ & $1$ & $3$\\
      \hline
      $0$ & $0$ & $*$ & $0$ & $1$ & $2$\\
      \hline
      $1$ & $0$ & $*$ & $0$ & $1$ & $3$\\
      \hline
      $0$ & $1$ & $0$ & $*$ & $0$ & $2$\\
      \hline
      $1$ & $1$ & $0$ & $*$ & $0$ & $3$\\
      \hline
      $0$ & $1$ & $1$ & $0$ & $*$ & $2$\\
      \hline
      $0$ & $1$ & $1$ & $0$ & $*$ & $3$\\
      \hline\hline
    \end{tabular}
    \caption{The table for the game with a special player.}
    \label{fig:5code4}
  \end{subfigure}
  \caption{Translation from measurement operators to game
    specifications.}
\end{figure}

The table in Fig.~\ref{fig:5code4} specifies a five-player game in
Fig.~\ref{fig:spsgame} called the {\it Special-Player Stabilizer
  Game\/}. In the game, the verifier randomly selects four players
each time, asks a question encoded with a single bit and expects a
single bit answer. He accepts or rejects depending on the parity of
the answers received. The fifth player is special because of the
$\pi/4$-rotation applied on the fifth column of the stabilizer
generators. This breaks the translation invariance of the five-qubit
code and, in the honest strategy, the fifth player performs
differently from other players.

\begin{figure}[!htb]
  \begin{framed}
    \ul{Special-Player Stabilizer Game}\\[1em]
    Let the $s_j$ for $j\in [8]$ be the eight entries in the parity
    column of Fig.~\ref{fig:5code4} and let $w_j = (w_{j,i})$ be the
    $j$-th row of the question column. The verifier does the
    following:
    \begin{enumerate}\setlength{\itemsep}{0pt}
    \item Select an index $j\in [8]$ uniformly at random.
    \item For $i\in [5]$, send $w_{j,i}$ in Fig.~\ref{fig:5code4} to
      the player $(i)$ if $w_{j,i}$ is not $*$, the null question.
    \item Receive a bit $a^{(i)}$ from player $(i)$ if she was asked a
      question.
    \item Accept if and only if the parity of the answers
      $\bigoplus_i a^{(i)}$ equals $s_j$.
    \end{enumerate}
    \caption{Special-player stabilizer game for the five-qubit code.
      The fifth player is the special player.}
    \label{fig:spsgame}
  \end{framed}
\end{figure}

It turns out that the special-player stabilizer game in
Fig.~\ref{fig:spsgame} has non-local value
\begin{equation}
  \label{eq:spsvalue}
  \omega^*_{\text{SPS}} = \frac{2+\sqrt{2}}{4},
\end{equation}
the same as that of the CHSH game. The following strategy achieves the value.
The five players share the state $\ket{\psi'}$ as in
Eq.~\eqref{eq:psi'} and measure $X$ (or $Z$) if the question is even
(or odd, respectively) and reply with the outcome. Let $h_{j,i}$ be
the $i$-th Pauli operator of $h_j$. The value this strategy achieves
is
\begin{equation}
  \label{eq:spscomp}
  \begin{split}
    \omega^*_{\text{SPS}} \, & = \E_j \frac{1+(-1)^{s_j}\bigbra{\psi'}
      \bigl( \bigotimes_{i=1}^5 h_{j,i} \bigr) \bigket{\psi'}}{2}\\
    & = \frac{1+\E_j \bigbra{\psi'} h_j \bigket{\psi'}}{2},
  \end{split}
\end{equation}
which gives the desired value using Eq.~\eqref{eq:h_i}. It is
beneficial to restate the above optimal strategy using $\ket{\psi}$
instead of $\ket{\psi'}$. By the conjugation relation between $X$, $Z$
and $X'$ and $Z'$, the fifth player essentially measures $X'$ and $Z'$
on the state $\ket{\psi}$ in the code space when the question is $2,3$
respectively. If we think of questions $0,1,2,3$ as measurement
specifications of $X$, $Z$, $X'$, $Z'$, then players who honestly
follow the measurement instructions on an encoded state has acceptance
probability $\omega^*_{\text{SPS}}$. We mention without proof that the
classical value of the game is $3/4$, again the same as that of the
CHSH game.

We have seen a strategy for the game with value $(2+\sqrt{2})/4$. The
fact that this value is optimal is given in the following theorem. The
theorem also proves an important structural result about strategies
that almost achieve the non-local value of the game. Namely, the
special player (the fifth player) must measure honestly the $X'$ and
$Z'$ measurements up to an isometry. It is a partial rigidity property
of the special-player stabilizer games.

\begin{theorem}
  \label{thm:sps}
  Let $\S = \bigl(\rho, \bigl\{ R^{(i)}_w \bigr\} \bigr)$ be a
  strategy for the special-player stabilizer game in
  Fig.~\ref{fig:spsgame} for the five-qubit code, where $\rho$ is the
  state shared between the players before the game starts and
  $R^{(i)}_w$ is the reflection on Hilbert space $\H_i$ describing the
  projective measurements the player $(i)$ performs when receiving
  question $w\in \{0,1,2,3\}$. Then the value of the strategy $\S$ is
  at most $\omega^*_{\text{SPS}}$ given in Eq.~\eqref{eq:spsvalue}.
  Furthermore, if the value is at least
  $\omega^*_{\text{SPS}} - \epsilon$, then there exists an isometry
  $V\in \Lin(\H_5,\B\otimes \hat{\H}_5)$ such that
  $R^{(5)}_3 = V^\dagger (Z\otimes I) V$, and
  \begin{equation*}
    d_\rho \bigl( R^{(5)}_2, V^\dagger (X\otimes I) V \bigr) \le
    O(\sqrt{\epsilon}).
  \end{equation*}
\end{theorem}

We note that some previous works use different distance measures for
measurements in the statement of rigidity theorem. For example, the
quantity $\norm{(R-S) \otimes I \ket{\psi}}$ is used in~\cite{RUV13}
for reflections $R$, $S$ on state $\ket{\psi}$. It is easy to verify
that this is the same as our distance measure $d_\rho(R,S)$ up to a
constant for $\rho = \ket{\psi}\bra{\psi}$.

The proof of the above theorem relies on the following lemmas.

\begin{lemma}[Jordan's Lemma~\cite{Jor75}]
  For any two reflections $R_0,R_1$ acting on a finite dimensional
  Hilbert space $\H$, there exists a decomposition of $\H$ into
  orthogonal one- and two-dimensional subspaces invariant under both
  $R_0$ and $R_1$.
\end{lemma}

\begin{lemma}
  \label{lem:abs}
  Let $R$ be a reflection and $H$ be a Hermitian matrix. Then
  \begin{equation*}
    \abs{\tr_\rho(R \otimes H)} \le \tr_\rho \abs{H}.
  \end{equation*}
\end{lemma}

\begin{proof}
  We first prove that
  \begin{equation*}
    \tr_\rho(R \otimes H) \le \tr_\rho \abs{H}.
  \end{equation*}
  Let $H^+$ and $H^-$ be the positive and negative part of
  $H = H^+ - H^-$ for positive semidefinite $H^+$ and $H^-$. The claim
  follows by a direct calculation
  \begin{align*}
    & \tr_\rho(R\otimes H) - \tr_\rho \abs{H}\\
    = & \tr_\rho(R\otimes (H^+-H^-)) - \tr_\rho(I\otimes(H^++H^-))\\
    = & -\tr_\rho((I-R) \otimes H^+) - \tr_\rho ((I+R)\otimes H^-)
        \le 0.
  \end{align*}
  A similar argument establishes
  \begin{equation*}
    \tr_\rho(R \otimes H) \ge - \tr_\rho \abs{H},
  \end{equation*}
  which completes the proof.
\end{proof}

\begin{lemma}
  \label{lem:ClSl}
  For $\theta_l \in [0,\pi]$, $C_l = \cos\theta_l$,
  $S_l = \sin\theta_l$, and any probability distribution over $l$, if
  \begin{equation*}
    \E_l \bigl( \sqrt{1+C_l}-1 \bigr)^2 \le \epsilon,
  \end{equation*}
  then
  \begin{equation*}
    \E_l S_l \ge 1 - O(\epsilon).
  \end{equation*}
\end{lemma}

\begin{proof}
  Let $\epsilon_l$ be $(\sqrt{1+C_l}-1)^2$. Then
  $\E_l \epsilon_l \le \epsilon$. For each $l$, $\epsilon_l \le 1$. We
  claim that for all $l$, $S_l \ge 1-9 \epsilon_l$. This will
  obviously finish the proof.

  As
  \begin{equation*}
    -2\sqrt{\epsilon_l} \le C_l = \epsilon_l \pm 2\sqrt{\epsilon_l}
    \le 3 \sqrt{\epsilon_l},
  \end{equation*}
  it follows that $C_l^2 \le 9 \epsilon_l$. If $\epsilon_l \ge 1/9$,
  the claim is trivial since $\theta_l \in [0,\pi]$ and $S_l \ge 0$.
  Otherwise,
  \begin{equation*}
    S_l = \sqrt{1-C_l^2} \ge \sqrt{1-9\epsilon_l} \ge 1-9\epsilon_l.
  \end{equation*}
\end{proof}

\begin{proof}[Proof of Theorem~\ref{thm:sps}]
  We first give an expression of the game value for the strategy $\S$.
  For each operator $h_j$ in Fig.~\ref{fig:5code3}, define $h_j(\S)$
  as the operator obtained by substituting the $X,Z$ operators in
  $h_j$ with the players' corresponding reflections in the strategy
  $\S$. That is
  \begin{equation*}
    h_j(\S) = (-1)^{s_j} \bigotimes_{i=1}^5 R^{(i)}_{w_{j,i}},
  \end{equation*}
  where $s_j$, $w_{j,i}$ are defined in Fig.~\ref{fig:spsgame} and
  $R^{(i)}_*$ is defined to be $I$. Following the similar steps as in
  Eq.~\eqref{eq:spscomp}, the value of $\S$ is computed as
  \begin{equation*}
    \omega^*(\S) = \frac{1+\E_j \tr_\rho (h_j(\S))}{2}.
  \end{equation*}

  Consider four matrices
  \begin{subequations}
    \label{eq:Delta}
    \begin{align}
      \Delta_1(\S) & = I\otimes R^{(2)}_0 \otimes R^{(3)}_1 \otimes
                     R^{(4)}_1 \otimes I - I^{\otimes 4} \otimes
                     \frac{R^{(5)}_2+R^{(5)}_3}{\sqrt{2}},\\
      \Delta_2(\S) & = R^{(1)}_0 \otimes I \otimes R^{(3)}_0 \otimes
                     R^{(4)}_1 \otimes I - I^{\otimes 4}  \otimes
                     \frac{R^{(5)}_2-R^{(5)}_3}{\sqrt{2}},\\
      \Delta_3(\S) & = R^{(1)}_1 \otimes R^{(2)}_0 \otimes I\otimes
                     R^{(4)}_0 \otimes I - I^{\otimes 4}  \otimes
                     \frac{R^{(5)}_2-R^{(5)}_3}{\sqrt{2}},\\
      \Delta_4(\S) & = R^{(1)}_1 \otimes R^{(2)}_1 \otimes R^{(3)}_0
                     \otimes I \otimes I - I^{\otimes 4} \otimes
                     \frac{R^{(5)}_2+R^{(5)}_3}{\sqrt{2}}.
    \end{align}
  \end{subequations}
  In the following, we write them as $\Delta_l$ and hide their
  dependence on the strategy $\S$ when there is no ambiguity. The sum
  of square of the four matrices is
  \begin{equation*}
    \sum_{l=1}^4 \Delta_l^2 = 8 I - \sqrt{2} \sum_j h_j(\S).
  \end{equation*}
  This gives the following expression for the game value
  \begin{equation*}
    \omega^*(\S) = \frac{1}{2} + \frac{8\sqrt{2} -
      \sqrt{2}\sum_{l=1}^4 \tr_\rho(\Delta^2_l)}{32},
  \end{equation*}
  from which the optimality of $\omega^*_{\text{SPS}}$ is obvious. It
  also implies that for any strategy $\S$ having value at least
  $\omega^*_{\text{SPS}} - \epsilon$,
  \begin{equation}
    \label{eq:D_bound}
    \tr_\rho (\Delta_1^2) \le \sum_{l=1}^4 \tr_\rho (\Delta_l^2) \le
    O(\epsilon).
  \end{equation}

  For simplicity, let $R_2, R_3$ be the shorthand notion of
  $R^{(5)}_2$ and $R^{(5)}_3$ in the rest of the proof. Following a
  similar truncation argument as in~\cite{RUV13}, we may assume
  without loss of generality that the underlying Hilbert spaces of the
  players are finite dimensional so that Jordan's Lemma applies. Using
  Jordan's lemma and adding extra dimensions if necessary, one get
  simultaneous $2$-by-$2$ block diagonalizations of $R_2$ and $R_3$
  such that each $2$-by-$2$ block is a reflection having both $\pm 1$
  eigenvalues. Hence, there is an isometry
  $V\in \Lin(\H_5, \B\otimes \hat{\H}_5)$ such that
  \begin{equation*}
    R_3 = V^\dagger (Z\otimes I) V,
  \end{equation*}
  and
  \begin{equation*}
    R_2 = V^\dagger \sum_l \left[
      \begin{pmatrix}
        C_l & \phantom{-}S_l\\
        S_l & -C_l
      \end{pmatrix}
      \otimes \op{l}{l}\right] V,
  \end{equation*}
  where $C_l = \cos\theta_l, S_l = \sin\theta_l$ for
  $\theta_l \in [0,\pi]$ and $l$ is the index of the two-dimensional
  invariant subspaces obtained by Jordan's lemma.

  Substitute the expression for $R_2$ and $R_3$ in
  Eq.~\eqref{eq:D_bound},
  \begin{equation}
    \label{eq:D_eps}
    \begin{split}
      O(\epsilon) \ge \tr_\rho (\Delta_1^2) & = 2 + \frac{1}{2}
      \tr_\rho (R_2R_3+R_3R_2) + \sqrt{2} \tr_\rho\bigl(
      R\otimes(R_2 + R_3)\bigr)\\
      & \ge 2 + \frac{1}{2} \tr_\rho
      (R_2R_3+R_3R_2)-\sqrt{2} \tr_\rho\abs{R_2+R_3}\\
      & = 2 + \frac{1}{2}\tr_{\tilde{\rho}} \left[\sum_l
        \begin{pmatrix}
          2C_l & 0\\
          0 & 2C_l
        \end{pmatrix}
        \otimes \op{l}{l} \right] - 2 \tr_{\tilde{\rho}} \bigl(
      \sum_l \sqrt{1+C_l} I \otimes \op{l}{l} \bigr)\\
      & = \E_l (\sqrt{1+C_l}-1)^2,
    \end{split}
  \end{equation}
  where $R$ is the reflection
  $I\otimes R^{(2)}_0 \otimes R^{(3)}_1 \otimes R^{(4)}_1$, the second
  line follows from Lemma~\ref{lem:abs},
  $\tilde{\rho} = V\rho V^\dagger$, and the expectation $\E_l$ is over
  the probability distribution
  $\Pr(l) = \tr_{\tilde{\rho}} (I\otimes \op{l}{l})$.

  To complete the proof, consider the state dependent distance between
  reflections $R_2$ and $V^\dagger (X\otimes I) V$
  \begin{equation*}
    \begin{split}
      d_\rho \bigl( R_2, V^\dagger (X\otimes I) V \bigr) & = \bigl[1 -
      \Re \tr_\rho \bigl( R_2 V^\dagger(X\otimes I) V \bigr)
      \bigr]^{\frac{1}{2}}\\
      & = \bigl(1 - \E_l S_l\bigr)^{\frac{1}{2}}.
    \end{split}
  \end{equation*}
  This equation and Eq.~\eqref{eq:D_eps} together with
  Lemma~\ref{lem:ClSl} give the second part in the theorem.
\end{proof}

In the above discussion, the fifth player plays the role of the
special player in the game. It is natural to generalize this to a game
with player $(t)$ as the special player. The five-qubit code game with
special player $(t)$ is the game defined by the table in
Fig.~\ref{fig:5code4} after a cyclic rotation of the question columns
such that the special column becomes the $t$-th one. This makes use of the
translation invariance of the five-qubit code.

For a general $r$-qubit stabilizer $\Stab$ that has a set of $XZ$-form
generators $g_1, g_2, \ldots, g_l$, we need to select two generators
$g^{(t)}_X$, $g^{(t)}_Z$ among $g_1, \ldots, g_l$ for each qubit
$t\in [r]$ such that $g^{(t)}_X$ has $X$ operator on the $t$-th qubit
and $g^{(t)}_Z$ has $Z$ operator on the $t$-th qubit. This is possible
in most cases as long as the distance of the stabilizer is larger than
or equal to $2$ and the $t$-th qubit is not fixed to a pure state for
all code states. We call a stabilizer non-trivial if such choices of a
pair of generators are possible for all $t\in [r]$. The choice of
$g^{(t)}_X$ and $g^{(t)}_Z$ may not be unique, but any such choice
will work. To play the special player stabilizer game with special
player $(t)$, we follow the procedure in Fig.~\ref{fig:spsgame} with
generators $g^{(t)}_X$ and $g^{(t)}_Z$. Even though the game may not
essentially depend on all the generators, the proof of partial
rigidity still follows in this general case.

More specifically, for $t\in [r]$, let $g^{(t)}_X$ and $g^{(t)}_Z$ be
generators that have $X$ and $Z$ on the $t$-th qubit respectively.
Following the idea in the design of the game for the five-qubit code,
define four operators obtained by doing the $\pi/4$-trick on the
$t$-th qubit to this pair of generators
\begin{equation}
  \label{eq:spsh}
  h^{(t)}_1 = g^{(t)}_X, \quad h^{(t)}_2 = g^{(t)}_{X\mapsto Z}, \quad
  h^{(t)}_3 = g^{(t)}_{Z\mapsto X}, \quad h^{(t)}_4 = -g^{(t)}_Z,
\end{equation}
where $g^{(t)}_{X\mapsto Z}$ is the operator obtained by changing the
$X$ operator of the $t$-th qubit of $g^{(t)}_X$ to $Z$ and, similarly,
$g^{(t)}_{Z\mapsto X}$ is the operator obtained by changing the $Z$ on
$t$-th qubit of $g^{(t)}_Z$ to $X$.

As in the case for the five-qubit code game, we translate the
operators $h^{(t)}_j$ to questions $w_j = (w_{j,i})$ where $w_{j,i}$
is in $ \{*, 0, 1, 2, 3\}$. Following the translation rule, $w_{j,i}$
is $*$ if the $i$-th tensor factor of $h_j$ is $I$. It is $0$, $1$ if
$i\ne t$ and the corresponding tensor factor is $X$, $Z$ respectively,
and it is $2$, $3$ if $i=t$ and the tensor factor is $X$, $Z$
respectively. Similarly, define $s_j$ to be $0$ or $1$ if the sign of
$h^{(t)}_j$ is $1$ or $-1$ respectively.

The stabilizer game with special player $(t)$ is the game defined as
in Fig.~\ref{fig:spsgen}.

\begin{figure}[!htb]
  \begin{framed}
    \ul{Special-Player Stabilizer Game (General Case)}\\[1em]
    For a stabilizer $\Stab$ with $XZ$-form generators, let $w_{j,i}$
    and $s_j$ be the questions and parities defined by the operators
    $h^{(t)}_j$ in Eq.~\eqref{eq:spsh} for $j\in [4]$. In the
    special-player stabilizer game with player $(t)$ as the special
    player, the verifier performs the following steps:
    \begin{enumerate}\setlength{\itemsep}{0pt}
    \item Select an index $j\in [4]$ uniformly at random.
    \item For $i\in [r]$, send $w_{j,i}$ in Fig.~\ref{fig:5code4} to
      the player $(i)$ if $w_{j,i}$ is not $*$, the null question.
    \item Receive a bit $a^{(i)}$ from player $(i)$ if she was asked a
      question.
    \item Accept if and only if the parity of the answers
      $\bigoplus_i a^{(i)}$ equals $s_j$.
    \end{enumerate}
    \caption{Special-player stabilizer game.}
    \label{fig:spsgen}
  \end{framed}
\end{figure}

\subsection{Stabilizer Game}

The partial rigidity of the special-player stabilizer game applies
only to the measurements preformed by the special player. It
essentially forces the special player to measure $X'$ and $Z'$ on her
system. There are, however, no rigidity known for the other players'
measurements and nothing is proved about the shared state of the
strategy. The stabilizer game uses the special-player stabilizer game
as a sub-module to achieve the full rigidity properties for all
players.

The specification of the {\it Stabilizer Game\/} is given in
Fig.~\ref{fig:sgame}. It is defined by the $r$-qubit stabilizer
$\Stab$ with $XZ$-form generators. It also implicitly depends on a
fixed choice of generators $g^{(t)}_X$, $g^{(t)}_Z$ for each
$t\in [r]$ in order to perform the second test. The game involves $r$
players, each of whom may receive a question of two bits, and is
required to answer one bit. The verifier's decision depends only on
the parity of some of the answer bits and the game is a generalized
XOR game. In this paper, the number of qubits $r$ of the stabilizer is
always assumed to be a constant, and it may come in the Big-$O$ notions
in the rest of the paper.

\begin{figure}[!htb]
  \begin{framed}
    \ul{Stabilizer Game}\\[1em]
    For an $r$-qubit non-trivial stabilizer $\Stab$ with $XZ$-form
    generators $g_1, g_2, \ldots, g_l$, define the stabilizer game of
    $\Stab$ as follows. Let $w_{j,i}$ for $j\in [l]$, $i\in [r]$ be
    $*$, $0$, or $1$ if $g_j$ has $I$, $X$, or $Z$ on the $i$-th qubit
    respectively. Let $s_j$ for $j\in [l]$ be the $0$, $1$ if the sign
    of $g_j$ is $1$, $-1$ respectively. The verifier performs the
    following two tests with equal probability:
    \begin{enumerate}\setlength{\itemsep}{0pt}
    \item Select an index $j\in [l]$ uniformly at random. Send
      $w_{j,i}$ to the player $(i)$ if $w_{j,i} \ne *$. Receives a bit
      from each player. Accept if the answers not corresponding to the
      null questions have the same parity as $s_j$. Reject otherwise.
    \item Select $t\in [r]$ uniformly at random. Play the
      special-player stabilizer game with special player $(t)$ in
      Fig.~\ref{fig:spsgen}.
    \end{enumerate}
    \caption{Stabilizer game for a stabilizer $\Stab$.}
    \label{fig:sgame}
  \end{framed}
\end{figure}

The non-local value of the game is
\begin{equation*}
  \omega^*_{\text{S}} = \frac{1 + \omega^*_\text{SPS}}{2} =
  \frac{6+\sqrt{2}}{8},
\end{equation*}
which can be achieved by players who share an encoded state of the
stabilizer code and measures $X$, $Z$, $X'$, $Z'$ when receiving
$0,1,2,3$, respectively. This value is easily seen to be optimal as it
saturates the winning probability in both tests of the stabilizer
game.

The stabilizer game has the following rigidity property.

\begin{theorem}
  \label{thm:sgame}
  For any non-trivial $r$-qubit stabilizer with $XZ$-form generators,
  the stabilizer game in Fig.~\ref{fig:sgame} has the following
  rigidity property. For any strategy $\S$ of the game specified by
  Hilbert spaces $\{ \H_i \}_{i=1}^r$, a state
  $\rho \in \Density(\bigotimes_{i=1}^r \H_i)$, and reflections
  $R^{(i)}_w$ on $\H_i$ for $i\in [r]$, $w\in \{0,1,2,3\}$, if the
  value of the strategy is at least $\omega^*_{\text{S}} - \epsilon$,
  then, there are isometries $V_i \in \Lin(\H_i,\B\otimes \hat{\H}_i)$
  for $i\in [r]$, such that the following properties holds
  \begin{itemize}
  \item For all $i\in[r]$, $R^{(i)}_3 = V_i^\dagger (Z'\otimes I) V_i$
    and
    \begin{subequations}
      \begin{align}
        d_\rho \bigl( R^{(i)}_2, V_i^\dagger (X'\otimes I) V_i \bigr)
        & \le O(\epsilon^{1/2}),\\
        d_\rho \bigl( R^{(i)}_1, V_i^\dagger (Z\otimes I) V_i \bigr)
        & \le O(\epsilon^{1/4}),\label{eq:sgame2}\\
        d_\rho \bigl(R^{(i)}_0, V_i^\dagger (X\otimes I) V_i \bigr)
        & \le O(\epsilon^{1/4}),\label{eq:sgame3}
      \end{align}
    \end{subequations}
    where $X', Z'$ are defined in Eq.~\eqref{eq:XZ'}.
  \item Let $\Pi$ be the projection to the code space of the
    stabilizer code of $\;\Stab$, and let $V$ be the isometry
    $\bigotimes_{i=1}^r V_i$, then
    \begin{equation*}
      \ip{\Pi\otimes I}{V\rho V^\dagger} \ge 1-O(\epsilon^{1/4}),
    \end{equation*}
    where $\Pi$ acts on the $r$ qubits, each of which is the first
    qubit of each player after the application of isometry $V$.
  \end{itemize}
\end{theorem}

\begin{proof}
  By the symmetry of the game, it suffices to prove the statement for
  one of the players, say, the player $(r)$. For simplicity, use $R_w$
  to represent the reflection $R^{(r)}_w$ of player $(r)$. It is easy
  to see that strategy $\S$ wins the special-player stabilizer game of
  special player $(r)$ with probability
  $\omega^*_{\text{SPS}} - O(\epsilon)$. By Theorem~\ref{thm:sps},
  there exists a isometry $V$ such that
  $R_3 = V^\dagger (Z\otimes I) V$ and
  \begin{equation*}
    d_\rho \bigl(R_2, V^\dagger (X\otimes I) V \bigr) \le
    O(\epsilon^{1/2}).
  \end{equation*}
  Taking $V_r = (W\otimes I) V$, we get the first two conditions in
  the first item of the theorem, where $W$ is the reflection defined
  in Eq.~\eqref{eq:W}.

  Next, we prove the claim in Eqs.~\eqref{eq:sgame2}
  and~\eqref{eq:sgame3}. For any $XZ$-form Pauli operator $g$, define
  $g(\S)$ to be the operator obtained by replacing $X$ with
  $R^{(i)}_0$ and $Z$ with $R^{(i)}_1$. Let $R$ be the product of the
  first $r-1$ tensor factors of $g^{(r)}_X(\S)$. Here, $g^{(r)}_X$ is
  the chosen generator that has $X$ on the $r$-th qubit in the
  stabilizer game with special player $(r)$. As the strategy $\S$ has
  value at least $\omega^*_{\text{S}} - \epsilon$, it has value at
  least $1-O(\epsilon)$ for the first test of the stabilizer game, and
  therefore,
  \begin{equation}
    \label{eq:consa}
    \tr_\rho (R\otimes R_0) = \tr_\rho \bigl( g^{(r)}_X(\S) \bigr) \ge
    1 - O(\epsilon).
  \end{equation}
  In other words, the reflection $R_0$ is $O(\epsilon)$-consistent
  with $R$ on $\rho$. We emphasize that the reflection $R$ acts on the
  joint system of the first $r-1$ players. This does not cause any
  problem as it is only used for our proof and is never actually
  measured on the joint system.

  Consider a new strategy $\hat{\S}$ modified from strategy $\S$ by
  changing $R_2$ in the strategy $\S$ to
  $V_r^\dagger (X'\otimes I) V_r$. This new strategy has value at
  least $\omega^*_{\text{S}} - O(\sqrt{\epsilon})$ by
  Lemma~\ref{lem:drho}. Consider the matrix $\Delta_1$ for strategy
  $\hat{\S}$, as in Eq.~\eqref{eq:Delta},
  \begin{equation*}
    \begin{split}
      \Delta_1(\hat{\S}) & = R \otimes I - I \otimes \Bigl[
      V_r^\dagger \bigl( \frac{X'+Z'}{\sqrt{2}} \otimes I \bigr) V_r
      \Bigr],\\
      & = R \otimes I - I \otimes V_r^\dagger (X\otimes I) V_r.
    \end{split}
  \end{equation*}
  By a similar argument that gives Eq.~\eqref{eq:D_bound} and the fact
  that $\hat{\S}$ has value at least
  $\omega^*_{\text{SPS}} - O(\sqrt{\epsilon})$ in the second part of
  the game, we have
  \begin{equation*}
    \tr_\rho (\Delta^2_1(\hat{\S})) \le O(\sqrt{\epsilon}).
  \end{equation*}
  This gives
  \begin{equation}
    \label{eq:consb}
    \tr_\rho (R\otimes V_r^\dagger (X\otimes I) V_r) \ge
    1-O(\sqrt{\epsilon}),
  \end{equation}
  which proves the $O(\sqrt{\epsilon})$-consistency of
  $ V_r^\dagger (X\otimes I) V_r$ and $R$ on $\rho$.

  The Eqs.~\eqref{eq:consa} and~\eqref{eq:consb} and
  Lemma~\ref{lem:cons} imply that
  \begin{equation*}
    d_\rho \bigl(R_0, V_r^\dagger (X\otimes I) V_r \bigr) \le
    O(\epsilon^{1/4}).
  \end{equation*}
  This completes the proof for Eq.~\eqref{eq:sgame3}. A similar
  argument establishes Eq.~\eqref{eq:sgame2}.

  Finally, to prove the second item of the theorem, consider a
  strategy $\tilde{\S}$ that uses the same state $\rho$ and
  reflections
  \begin{alignat*}{2}
    R^{(i)}_0 & = V_i^\dagger (X\otimes I) V_i, \qquad & R^{(i)}_2 & =
    V_i^\dagger(X'\otimes I) V_i,\\
    R^{(i)}_1 & = V_i^\dagger (Z\otimes I) V_i, \qquad & R^{(i)}_3 & =
    V_i^\dagger(Z'\otimes I) V_i.
  \end{alignat*}
  By Lemma~\ref{lem:drho} and the first part of the theorem, strategy
  $\tilde{\S}$ has value at least $\omega^*_S - O(\epsilon^{1/4})$.
  Hence, it has acceptance probability at least $1-O(\epsilon^{1/4})$
  in the first test of the stabilizer game. This means that
  \begin{equation*}
    \frac{1+\E_j \tr_{\tilde{\rho}} (g_j) }{2} = 1 - O(\epsilon^{1/4}),
  \end{equation*}
  where $j$ is uniformly random over $[l]$, $g_j$'s are generators of
  the stabilizer, and $\tilde{\rho} = V^\dagger \rho V$. This is
  equivalent to
  \begin{equation}
    \label{eq:sgame4}
    \tr_{\tilde{\rho}} \bigl( \sum_{j=1}^l g_j \bigr) = l - O(\epsilon^{1/4}).
  \end{equation}
  Operator $\sum_{j=1}^l g_j$ has eigenvalues in
  $\{-l, -l+2, \ldots, l-2, l\}$ and $\Pi$ projects to the eigenspace
  of eigenvalue $l$. Hence
  \begin{equation*}
    \sum_{j=1}^l g_j \le l\Pi + (l-2)(I-\Pi) = (l-2)I + 2\Pi.
  \end{equation*}
  This, together with Eq.~\eqref{eq:sgame4}, implies that
  \begin{equation*}
    \tr_{\tilde{\rho}} (\Pi) = 1 - O(\epsilon^{1/4}),
  \end{equation*}
  which is equivalent to the second part of the theorem.
\end{proof}

\subsection{Multi-Qubit Stabilizer Game}

In this section, we consider a multi-qubit variant of the stabilizer
game called the $(k,n)$-stabilizer game. It is a non-local game
implementation of the stabilizer check of the Fitzsimons-Vidick
protocol~\cite{FV15}. Instead of asking for the qubits and performing
encoding check on them, the verifier sends the measurement
instructions on the corresponding qubits to the players. The optimal
strategy of the game is to encode each qubit with the stabilizer code
and measure honestly the $X$, $Z$, $X'$ and $Z'$ on the encoded data
on corresponding qubits. We prove a partial rigidity theorem for the
multi-qubit stabilizer game, which suffices for our purpose. In
particular, we only prove the rigidity for the reflections
corresponding to questions in $0,1$. The full rigidity properties can
be proved with some little extra effort.

The $(k,n)$-stabilizer game is given in Fig.~\ref{fig:mqsgame}. For
simplicity, we assume in the multi-qubit stabilizer game that, when
the question is $*$, the verifier replaces it with either $0$ or $1$
and ignore the corresponding answer. With this convention, each player
will either see a question of the form $(u,w)$ for $u\in [n]$ and
$w=0,1,2,3$ or a tuple of $k$ such questions. Answers are either a
single bit or a string of $k$-bits correspondingly. The verifier
accepts or rejects depending on the parity of some of the answer bits.

\begin{figure}[!htb]
  \begin{framed}
    \ul{Multi-Qubit Stabilizer Game}\\[1em]
    Let $\Stab$ be a non-trivial $r$-qubit stabilizer with a set of
    generators of $XZ$-form. Let $[n]$ be the index of $n$ qubits and
    let $k\ge 2$ be a constant. The $(k,n)$-stabilizer game for
    $\Stab$ is an $r$-player non-local game where the verifier does
    the following with equal probability:
    \begin{enumerate}\setlength{\itemsep}{0pt}
    \item Select a subset $J \subset [n]$ of size $k$, an index
      $u\in J$, and a player $t\in [r]$, all uniformly at random. For
      each qubit $v\in J$, randomly select questions
      $w_v = (w_{v,i})$, where $w_{v,i} \in [0, 1, 2, 3]$ and each
      $w_v$ is sampled as in the stabilizer game. Define
      $q^{(i)}_v = (v,w_{v,i})$ for $i\in [r]$ and $v\in J$. Send
      $q^{(i)}_u$ to player $(i)$ and receive an answer bit $a^{(i)}$
      if $i\ne t$. Send $\vq = \bigl( q^{(t)}_v \bigr)_{v\in J}$ to
      player $(t)$, and receive a $k$-bit string $b=(b_v)_{v\in J}$.
      Define $a^{(t)} = b_u$ and
      $a=(a^{(1)}, a^{(2)}, \ldots, a^{(r)})$. The verifier accepts if
      and only if the verifier for the stabilizer game accepts when
      the questions are $w_u$ and answers are $a$.
    \item Select a qubit $u\in [n]$ uniformly at random. Play the
      stabilizer game on qubit $u$. That is, the verifier sample
      $w = (w_i)$ as in the stabilizer game. Define $q^{(i)}=(u,w_i)$
      for $i\in [r]$. Send $q^{(i)}$ to player $(i)$ and receive an
      answer bit $a^{(i)}$. The verifier accepts if the verifier for
      the stabilizer game accepts on questions $w$ and answers
      $a=(a^{(i)})$.
    \end{enumerate}
    \caption{Multi-qubit stabilizer game.}
    \label{fig:mqsgame}
  \end{framed}
\end{figure}

It is easy to see that the non-local value $\omega^*_{\text{MQS}}$ of
the $(k,n)$-stabilizer game in Fig.~\ref{fig:mqsgame} equals to the
value of the stabilizer game $\omega^*_{\text{S}}$. Let $\H_i$ be the
state space of player $(i)$. A strategy for the $k$-qubit stabilizer
game,
\begin{equation*}
  \S = \bigl(\rho, \bigl\{ R^{(i)}_q \bigr\}, \bigl\{
  M^{(i)}_\vq \bigr\} \bigr),
\end{equation*}
consists of a state
$\rho \in \Density\bigl( \bigotimes_{i=1}^r \H_i \bigr)$, reflections
$R^{(i)}_q$ the players measure for question $q$ and measurements
$M^{(i)}_\vq$ with $k$-bit outcomes for question $\vq$. The
superscripts of the measurements indexing the players are sometimes
omitted if there will be no ambiguity.

Without loss of generality, it is assumed that the measurements
$M_\vq$ are projective measurements. For each $q$ that occurs as the
$i$-th entry in the tuple $\vq$, define a reflection
\begin{equation*}
  S_{q|\vq} = \sum_{b\in \{0,1\}^k} (-1)^{b_i} M^b_\vq.
\end{equation*}
For $\vq=(q_1,q_2,\ldots,q_k)$, the measurement $M_\vq$ has a
one-to-one correspondence with the collection of $k$ pairwise
commuting reflections
\begin{equation*}
  \bigl\{ S_{q_s|\vq} \bigr\}_{s=1}^k,
\end{equation*}
and we refer to these reflections as the reflections associated with
the projective measurement $M_\vq$. We also write $S_{q|\vq}$ as
$S_{u,w|\vq}$ for $q=(u,w)$. On the other hand, for any collection of
$k$ pairwise commuting reflections, there associates a projective
measurement with $k$-bit outcome as the repeated application of the
two-outcome measurements defined by the reflections.

We prove the following partial rigidity property of the
$(k,n)$-stabilizer game.

\begin{theorem}
  \label{thm:mqsgame}
  For any constant integer $k \ge 2$, there exists a constant
  $\kappa>0$ that depends only on $k$ such that the $(k,n)$-stabilizer
  game in Fig.~\ref{fig:mqsgame} has the following rigidity property.
  For any quantum strategy
  $\S = \bigl(\rho, \bigl\{ R^{(i)}_q \bigr\}, \bigl\{ M^{(i)}_\vq
  \bigr\} \bigr)$
  that has value at least $\omega^*_{\text{S}} - \epsilon$, there are
  isometries $V_i \in \Lin(\H_i,\B^{\otimes n} \otimes \hat{\H}_i)$,
  such that the following properties hold
  \begin{itemize}
  \item For all $i\in[r]$, all $q=(u,w)$, $q_s = (u_s,w_s)$ with
    $u, u_s\in [n], w, w_s \in\{0,1\}, s\in[k]$, and
    $\vq=(q_1,q_2,\ldots,q_k)$,
    \begin{subequations}
      \begin{align}
        d_\rho \bigl( R^{(i)}_q, V_i^\dagger D_q V_i \bigr)
        & \le O(n^\kappa \epsilon^{1/\kappa}),\label{eq:mqs1}\\
        d_\rho \bigl( M^{(i)}_\vq, N^{(i)}_\vq \bigr)
        & \le O(n^\kappa \epsilon^{1/\kappa}),\label{eq:mqs2}
      \end{align}
    \end{subequations}
    where $D_q$ is the $X$, $Z$ operator on the $u$-th qubit for
    $w=0,1$ respectively, and $N^{(i)}_{\vq}$ is the projective
    measurement with $k$-bit outcome associated with the reflections
    $\bigl\{ V_i^\dagger D_{q_s} V_i \bigr\}_{s=1}^k$.
  \item Let $\Pi$ be the projection to the code space of the
    stabilizer code, $V$ be the isometry $\bigotimes_{i=1}^r V_i$,
    then
    \begin{equation}
      \label{eq:mqs3}
      \ip{\Pi^{\otimes n}\otimes I}{V\rho V^\dagger} \ge
      1-O(n^\kappa \epsilon^{1/\kappa}),
    \end{equation}
    where the $t$-th tensor factor of $\; \Pi^{\otimes n}$ acts on $r$
    qubits, each of which is the $t$-th qubit of each player's system
    after the application of $V$.
  \end{itemize}
\end{theorem}

The proof of Theorem~\ref{thm:mqsgame} relies on the following lemmas.

\begin{lemma}
  \label{lem:mqscomm}
  Let $\rho\in \Density(\H_A\otimes \H_B)$ be a state on systems $A$
  and $B$. Let $M_0, M_1, N_0, N_1$ be four projective measurements on
  $\H_A$ such that $M_1^a,N_1^a$ commute for all $a$. Let $M, N$ be
  two projective measurements on $\H_B$. Suppose that $M_0, M_1$ are
  both $\epsilon$-consistent with $M$, and $N_0, N_1$ are both
  $\epsilon$-consistent with $N$ on state $\rho$. Then
  \begin{equation*}
    \sum_a \norm{[M_0^a,N_0^a]}_\rho^2 \le O(\sqrt{\epsilon}).
  \end{equation*}
\end{lemma}

\begin{proof}
  First prove that
  \begin{equation*}
    \sum_a \tr_\rho \bigl( M_0^a N_0^a M_0^a N_0^a \bigr)
    \approx_{\sqrt{\epsilon}} \sum_a \tr_\rho \bigl( N_0^a M_0^a N_0^a
    M_1^a \bigr).
  \end{equation*}
  Namely, we can move operator $M_0^a$ in the front to the end of the
  product and change it to $M_1^a$ without incurring too much error in
  the expression.

  By $\epsilon$-consistency between $M_0$ and $M$,
  \begin{equation*}
    \sum_a \tr_\rho \bigl( M_0^a N_0^a M_0^a N_0^a \bigr)
    \approx_{\sqrt{\epsilon}} \sum_a \tr_\rho \bigl( M_0^a N_0^a M_0^a
    N_0^a \otimes M^a \bigr),
  \end{equation*}
  as the absolute value of the difference on the two sides is
  \begin{equation*}
    \begin{split}
      & \abs{\sum_a \tr_\rho \bigl( M_0^a N_0^a M_0^a N_0^a \otimes
        (1-M^a) \bigr)}\\
      \le & \sqrt{\sum_a \tr_\rho \bigl( M_0^a \otimes (1-M^a)\bigr)
        \sum_a \tr_\rho \bigl( N_0^a M_0^a N_0^a M_0^a N_0^a \bigr)}\\
      \le & \sqrt{\epsilon}.
    \end{split}
  \end{equation*}
  By similar arguments,
  \begin{equation*}
    \begin{split}
      \sum_a \tr_\rho \bigl( M_0^a N_0^a M_0^a N_0^a \otimes M^a
      \bigr) & \approx_{\sqrt{\epsilon}}
      \sum_a \tr_\rho \bigl( N_0^a M_0^a N_0^a \otimes M^a \bigr)\\
      & \approx_{\sqrt{\epsilon}} \sum_a \tr_\rho \bigl( N_0^a M_0^a
      N_0^a M_1^a \otimes M^a
      \bigr)\\
      & \approx_{\sqrt{\epsilon}} \sum_a \tr_\rho \bigl( N_0^a M_0^a
      N_0^a M_1^a \bigr).
    \end{split}
  \end{equation*}
  This proves our claim about moving and changing $M_0^a$ to $M_1^a$.
  Do this for the four operators in the product sequentially,
  \begin{equation*}
    \sum_a \tr_\rho \bigl( M_0^a N_0^a M_0^a N_0^a \bigr)
    \approx_{\sqrt{\epsilon}} \sum_a \tr_\rho \bigl( M_1^a N_1^a M_1^a
    N_1^a \bigr).
  \end{equation*}

  Expand the $\rho$-norm in left hand side of the claim in the lemma,
  \begin{equation*}
    \sum_a \norm{[M_0^a,N_0^a]}_\rho^2 = \sum_a \tr_\rho \bigl[M_0^a
    N_0^a M_0^a + N_0^a M_0^a N_0^a - M_0^a N_0^a M_0^a N_0^a  - N_0^a
    M_0^a N_0^a M_0^a \bigr].
  \end{equation*}
  Now the proof follows by applying the above operator moving
  procedure to each of the four terms in the expansion and the
  condition that $M_1^a$ commutes with $N_1^a$ for all $a$.
\end{proof}

\begin{lemma}
  \label{lem:C}
  Let $\B_1$, $\B_{1'}$, $\B$ be two-dimensional Hilbert spaces. Let
  $V\in \Lin(\H, \B \otimes \hat{\H})$ be an isometry,
  $R \in \Lin(\H)$ be an operator and $\ket{\Phi}$ be the EPR state on
  $\B_1\otimes \B_{1'}$. Define isometry
  $C\in \Lin(\H, \B_1\otimes \B_{1'}\otimes \H)$ as
  \begin{equation*}
    C = (I\otimes V^\dagger) \SWAP (\ket{\Phi} \otimes V),
  \end{equation*}
  where the $\SWAP$ acts on $\B_1$ and $\B$. Then
  \begin{equation*}
    C^\dagger R C = \E_{j=0,1,2,3} \bigl( V^\dagger \sigma_j V \bigr) R
    \bigl( V^\dagger \sigma_j V \bigr),
  \end{equation*}
  where $\sigma_j$ are Pauli operators.
\end{lemma}

\begin{proof}
  This follows by substituting the two $\SWAP$ gates using the
  identity
  \begin{equation*}
    \SWAP = \frac{I+XX+YY+ZZ}{2},
  \end{equation*}
  and a direct calculation.
\end{proof}

\begin{lemma}
  \label{lem:appstab}
  Let $\rho \in \Density(\H\otimes \H')$ be a state, $T\in \Lin(\H)$
  be an operator with constant operator norm, $R$ be a reflection on
  $\H$ that has an $\epsilon$-consistent reflection $S$ on $\H'$. Then
  \begin{equation*}
    \tr_\rho(RTR) \approx_{\sqrt{\epsilon}} \tr_\rho(T).
  \end{equation*}
\end{lemma}

\begin{proof}
  We first prove that
  \begin{equation}
    \label{eq:appstab1}
    \tr_\rho(RT) \approx_{\sqrt{\epsilon}} \tr_\rho(T \otimes S).
  \end{equation}

  By consistency of $R$ and $S$, we have
  \begin{equation*}
    \begin{split}
      \tr_\rho(RT) & = \sum_{a\in \{0,1\}} \tr_\rho(R^a (-1)^a T)\\
      & \approx_{\sqrt{\epsilon}} \sum_{a\in \{0,1\}} \tr_\rho(R^a
      (-1)^a T \otimes S^a)\\
      & \approx_{\sqrt{\epsilon}} \sum_{a\in \{0,1\}} \tr_\rho((-1)^a
      T \otimes S^a)\\
      & = \tr_\rho (T \otimes S).
    \end{split}
  \end{equation*}

  Similarly,
  \begin{equation}
    \label{eq:appstab2}
    \tr_\rho(TR) \approx_{\sqrt{\epsilon}} \tr_\rho(T \otimes S).
  \end{equation}

  Taking $T=TR$ in Eqs.~\eqref{eq:appstab1} and~\eqref{eq:appstab2},
  we have
  $\tr_\rho(RTR) \approx_{\sqrt{\epsilon}} \tr_\rho(TR \otimes S)$,
  and
  \begin{equation*}
    \tr_\rho(T) = \tr_\rho(TRR) \approx_{\sqrt{\epsilon}} \tr_\rho(TR
    \otimes S).
  \end{equation*}
  This completes the proof.
\end{proof}

\begin{lemma}
  \label{lem:repeat}
  Let $\rho \in \Density(\H_A\otimes \H_B)$ be a quantum state of
  systems $A$ and $B$, $M=\{ M^b \}$ and $N=\{ N^b \}$ be two
  projective measurements with $k$-bit outcomes on system $A$. Let
  $\{R_j\}_{j=1}^k$, $\{S_j\}_{j=1}^k$ be the reflections associated
  with the projective measurements $M$ and $N$ respectively. If there
  are reflections $T_j$ such that $R_j$ and $S_j$ are both
  $\epsilon$-consistent with $T_j$ on $\rho$ for all
  $j=1,2,\ldots, k$, then
  \begin{equation*}
    d_\rho(M, N) \le O(\epsilon^{1/4}).
  \end{equation*}
\end{lemma}

\begin{proof}
  We prove the case of $k=2$. The general case follows from a similar
  argument. Using the consistency conditions,
  \begin{equation*}
    \begin{split}
      \sum_{b\in \{0,1\}^2} \tr_\rho(M^bN^b) & = \sum_{b_1,b_2\in
        \{0,1\}}\tr_\rho(R^{b_1}_1
      R^{b_2}_2 S^{b_2}_2 S^{b_1}_1)\\
      & \approx_{\sqrt{\epsilon}} \sum_{b_1,b_2\in
        \{0,1\}}\tr_\rho(R^{b_1}_1
      R^{b_2}_2 S^{b_2}_2 S^{b_1}_1 \otimes T_1^{b_1})\\
      & \approx_{\sqrt{\epsilon}} \sum_{b_1,b_2\in \{0,1\}}\tr_\rho(
      R^{b_2}_2 S^{b_2}_2 \otimes T_1^{b_1})\\
      & = \sum_{b_2\in \{0,1\}}\tr_\rho( R^{b_2}_2 S^{b_2}_2)
      \approx_{\epsilon} 1.
    \end{split}
  \end{equation*}
  The first approximation uses the consistency condition to add the
  operator $T_1^{b_1}$, the second approximation uses the consistency
  condition to remove operators $R^{b_1}_1$ and $S^{b_1}_1$, the last
  approximation is by Lemma~\ref{lem:cons}.
\end{proof}

\begin{proof}[Proof of Theorem~\ref{thm:mqsgame}]

  Consider first the claim in Eq.~\eqref{eq:mqs1} of the theorem for
  $i=r$.

  In the proof, $D_w$ denotes $X$, $Z$ for $w=0, 1$ respectively, and
  $D_q$ denotes $D_w$ acting on the $u$-th qubit if $q=(u,w)$. Let
  $\delta$ be $n\epsilon$ and $\delta_k$ be $n^k\epsilon$. For
  simplicity, we omit the superscript $(r)$ in the reflections for
  player $(r)$.

  Since the strategy $\S$ has value at least
  $\omega^*_{\text{S}} - \epsilon$ for the $(k,n)$-stabilizer game, it
  must have value at least $\omega^*_{\text{S}} - O(\delta)$ for the
  stabilizer games for each $u\in [n]$ in the second part of the game.
  More precisely,
  $\S_u = \bigl( \rho, \bigl\{ R^{(i)}_{u,w} \bigr\} \bigr)$ forms a
  strategy for the stabilizer game with value at least
  $\omega^*_{\text{S}} - O(\delta)$.

  By Theorem~\ref{thm:sgame}, for all $u\in [n]$, there exist
  isometries $V_u \in \Lin(\H_r, \B \otimes \tilde{\H_r})$ such that
  \begin{equation*}
    d_\rho \bigl( R_{u,w}, V_u^\dagger (D_w\otimes I) V_u \bigr) \le
    O(\delta^{1/4}).
  \end{equation*}
  Define $\hat{R}_{u,w} = V_u^\dagger (D_w\otimes I) V_u$. The above
  equation becomes
  \begin{equation}
    \label{eq:R1}
    d_\rho \bigl( R_{u,w}, \hat{R}_{u,w} \bigr) \le O(\delta^{1/4}).
  \end{equation}

  Similarly, for all $J\subseteq [n]$ and $u\in J$, all choices of
  $w_v$ for $v\in J$ and $v\ne u$, consider state $\rho$, reflections
  $R^{(i)}_{u,w}$ for $i\in [r-1]$ and $S_{u,w|\vq}$ for player $(r)$.
  They form a strategy for the stabilizer game with value at least
  $\omega^*_{\text{S}} - O(\delta_k)$ for $q_u=(u,w)$, $q_v=(v,w_v)$
  and $\vq = (q_v)_{v\in J}$. We clarify that only $w$ is the index of
  questions for the stabilizer game in this strategy and everything
  else in the subscripts are fixed.

  By Eqs.~\eqref{eq:consa} and~\eqref{eq:consb}, reflection
  $S_{u,w|\vq}$ and $\hat{R}_{u,w}$ have the same
  $O(\sqrt{\delta_k})$-consistent measurement on $\rho$. Let $(u,w)$
  and $(v,w')$ be two entries of $\vq$, then the reflections
  $\hat{R}_{u,w}$, $S_{u,w|\vq}$, $S_{v,w'|\vq}$, $\hat{R}_{v,w'}$
  corresponds to measurements that satisfy the conditions for $M_0$,
  $M_1$, $N_1$, $N_0$ in Lemma~\ref{lem:mqscomm}. Hence,
  \begin{equation}
    \label{eq:Rcomm}
    \norm{\bigl[ \hat{R}_{u,w}, \hat{R}_{v,w'} \bigr]}_\rho^2 \le
    O(\delta_k^{1/4}),
  \end{equation}
  for all $u\ne v\in [n]$.

  Consider $2n$ two-dimensional Hilbert spaces $\B_u$, $\B_{u'}$ for
  $u\in [n]$. Denote $\H_{[n]} = \bigotimes_{u=1}^n \B_u$ and
  $\H_{[n]'} = \bigotimes_{u=1}^n \B_{u'}$. We sometimes call the
  system of $\B_u$ as the $u$-th qubit. Let $\ket{\Phi}_{u,u'}$ be the
  EPR state on systems $\B_u, \B_{u'}$. For each $u\in [n]$, define
  isometry $C_u \in \Lin(\H, \B_u\otimes \B_{u'} \otimes \H)$ as
  \begin{equation*}
    C_u = (I\otimes V_u^\dagger) \SWAP_u (\ket{\Phi}_{u,u'} \otimes V_u),
  \end{equation*}
  where $\SWAP_u$ is the $\SWAP$ gate acting on the $u$-th qubit and
  the first output qubit of $V_u$.

  Define isometry
  $V\in \Lin \bigl(\H_r, \H_{[n]} \otimes \H_{[n]'} \otimes \H_r
  \bigr)$ as the sequential application of $C_1, C_2, \ldots, C_n$,
  \begin{equation}
    \label{eq:V}
    V = C_n C_{n-1} \cdots C_1.
  \end{equation}
  We claim that this choice of $V$ works for the claims of the theorem
  by taking $\hat{\H}_r$ to be $\H_{[n]'} \otimes \H_r$. Define
  \begin{equation*}
    \tilde{R}_q = V^\dagger D_q V,
  \end{equation*}
  for $q=(u,w)$. The aim is to first prove that
  \begin{equation*}
    d_\rho(\hat{R}_q, \tilde{R}_q) \le O(n^\kappa \epsilon^{1/\kappa}).
  \end{equation*}

  Recall that $\sigma_s$ for $s=0,1,2,3$ are the Pauli operators.
  Define reflection $T_q = V_u^\dagger \sigma_s V_u$ for $q=(u,s)$ and
  the corresponding superoperator $\T_q(\sigma) = T_q \sigma T_q$. It
  is easy to verify that
  \begin{equation}
    \label{eq:T}
    T_{u,0} = I,\quad T_{u,1} = \hat{R}_{u,0},\quad T_{u,2} =
    -i\hat{R}_{u,0}\hat{R}_{u,1},\quad T_{u,3} = \hat{R}_{u,1}.
  \end{equation}

  As $D_q$ and $C_v$ commutes for all $v>u$ and $q=(u,w)$, we have
  \begin{equation*}
    \tilde{R}_q = C_1^\dagger C_2^\dagger \ldots C_{u-1}^\dagger
    \hat{R}_q C_{u-1} C_{u-2} \ldots C_1.
  \end{equation*}

  A series of applications of Lemma~\ref{lem:C} gives the expression
  of $\tilde{R}_q$ for $q=(u,w)$,
  \begin{equation*}
    \tilde{R}_q = \E_{s\in \{0,1,2,3\}^{u-1}} \T_{1,s_1} \circ
    \T_{2,s_2} \circ \cdots \circ \T_{u-1,s_{u-1}} (\hat{R}_q).
  \end{equation*}

  For convenience, define
  \begin{equation*}
    R = \E_{s_2,\ldots,s_{u-1}} \T_{2,s_2} \circ \cdots \circ
    \T_{u-1,s_{u-1}} (\hat{R}_q).
  \end{equation*}
  Then
  \begin{equation*}
    \tilde{R}_q = \E_{s_1 \in \{0,1,2,3\}} \T_{1,s_1} (R),
  \end{equation*}
  and
  \begin{equation*}
    \tr_\rho \bigl( \tilde{R}_q \hat{R}_q \bigr) = \E_{s_1 \in \{0,1,2,3\}}
    \tr_\rho \bigl( \T_{1,s_1} (R) \hat{R}_q \bigr).
  \end{equation*}

  For each of the four cases for $s_1$, it is claimed that
  $\tr_\rho \bigl( \T_{1,s_1} (R) \hat{R}_q \bigr)$ is close to
  $\tr_\rho (R\hat{R}_q)$. In words, removing the superoperator
  $\T_{1,s_1}$ induces a bounded error in the expression.

  Consider the case $s_1 = 1$ first. In this case
  \begin{equation*}
    \begin{split}
      \tr_\rho \bigl( \T_{1,1} (R) \hat{R}_q \bigr) & =
      \tr_\rho \bigl( \hat{R}_{1,0} R \hat{R}_{1,0} \hat{R}_q \bigr)\\
      & \approx_{\delta_k^{1/8}}
      \tr_\rho \bigl( \hat{R}_{1,0}R\hat{R}_q\hat{R}_{1,0} \bigr)\\
      & \approx_{\delta^{1/4}} \tr_\rho (R\hat{R}_q),
    \end{split}
  \end{equation*}
  where the first approximation follows from Eq.~\eqref{eq:Rcomm} and
  Cauchy-Schwarz inequality, the second approximation follows from
  Lemma~\ref{lem:appstab} and the fact that $\hat{R}_{1,0}$ has an
  $O(\delta^{1/2})$-consistency reflection on $\rho$ by
  Eq.~\eqref{eq:consb}.

  A similar argument applies for the other cases of $s_1$. Repeat this
  procedure of removing the superoperators $\T_{j,s_j}$ one by one, we
  have
  \begin{equation*}
    \tr_\rho (\tilde{R}_q,\hat{R}_q) \approx_{u\,\delta_k^{1/8}}
    \tr_\rho (\hat{R}_q,\hat{R}_q) = 1,
  \end{equation*}
  and
  \begin{equation*}
    d_\rho (\tilde{R}_q,\hat{R}_q) \le O\bigl( \sqrt{u
      \,\delta_k^{1/8}} \bigr) \le O(n^{1/2}\delta_k^{1/16}).
  \end{equation*}

  By the triangle inequality of $d_\rho$ and Eq.~\eqref{eq:R1}
  \begin{equation}
    \label{eq:tildeR}
    d_\rho (\tilde{R}_q,R_q) \le O(n^{1/2}\delta_k^{1/16}).
  \end{equation}
  This proves the bound in Eq.~\eqref{eq:mqs1} by choosing $\kappa$
  sufficiently large.

  Recall that there exists a reflection $R$ that is
  $\delta_k$-consistent with both $R_q$ and $S_{q|\vq}$ on $\rho$. By
  the bound in Eq.~\eqref{eq:tildeR} and Lemma~\ref{lem:drho},
  \begin{equation*}
    C_\rho (R,\tilde{R}_q) \ge 1 - O(n^{1/2} \delta_k^{1/16}).
  \end{equation*}
  The bounds in Eq.~\eqref{eq:mqs2} follows from
  Lemma~\ref{lem:repeat} and the consistency of reflections
  $\tilde{R}_q$, $S_{q|\vq}$ with the same reflection $R$ on the first
  $r-1$ players' systems.

  The second part of the theorem follows by a similar argument used to
  prove the second part of Theorem~\ref{thm:sgame}.
\end{proof}

\section{Non-Local Games for Local Hamiltonian Problems}

\label{sec:lh}

In this section, we give the non-local game for the local Hamiltonian
problem. Consider a restricted form of the local Hamiltonian
problem in the following definition.

\begin{definition}
  Consider a $k$-local Hamiltonian of $m$ terms on $n$ qubits
  \begin{equation*}
    H = \sum_{j=1}^m H_j,
  \end{equation*}
  where $0\le H_j\le I$ acts on at most $k$ qubits. The Hamiltonian
  $H$ is $XZ$-form if $H_j$ is a real linear combination of tensor
  products of $I,X,Z$ for each $j$.
\end{definition}

\begin{lemma}
  \label{lem:LH}
  There exist constant $k$ such that the $XZ$-form $k$-local
  Hamiltonian problem is \class{QMA}-complete.
\end{lemma}

\begin{proof}
  This is a simple corollary of the results in~\cite{BL08}. The gate
  set $\{ \CNOT, X, W = \cos(\pi/8)X + \sin(\pi/8)Z \}$ is known to be
  universal by the result of Shi~\cite{Shi02} and each gate $U_t$ in
  the gate set is of $XZ$-form and $U_t = U_t^\dagger$. Start with a
  circuit using this particular set of gates and perform the circuit
  to Hamiltonian construction of Kitaev. The $5$-local terms resulting
  from the construction will have the $XZ$-form. For example, the term
  checking the propagation of the $t$-th step of the circuit is
  \begin{equation*}
    \begin{split}
      H_{\text{prop},t} = & I\otimes \op{100}{100}_{t-1,t,t+1} - U_t
      \otimes
      \op{110}{100}_{t-1,t,t+1} \\
      & - U_t^\dagger \otimes \op{100}{110}_{t-1,t,t+1} + I\otimes
      \op{110}{110}_{t-1,t,t+1}\\
      = & \frac{I-Z_{(t-1)}}{2} \otimes \frac{I+Z_{(t+1)}}{2} - U_t
      \otimes \frac{I-Z_{(t-1)}}{2} \otimes X_{(t)} \otimes
      \frac{1+Z_{(t+1)}}{2},
    \end{split}
  \end{equation*}
  which is easily seen to have $XZ$-form for $U_t$ in the chosen gate
  set. Other terms in the construction can be checked similarly. This
  proves the claim for $k=5$.

  Using more advanced results from~\cite{CM14}, such as their Lemma
  22, one can prove the claim for $k=2$ and the two-local terms $H_j$
  have the form $\alpha_j (XZ-ZX)$.
\end{proof}

Let $H=\sum_{j=1}^m H_j$ be the Hamiltonian and assume that
$0\le H_j \le I$. We will consider two different types of energy
measurements for a $k$-local Hamiltonian $H$ on a state $\rho$. The
first type of measurement is the one used in~\cite{FV15} and does the
following. The verifier randomly selects $j\in[m]$ and gets the
$k$-qubit state $\rho_j$ on which $H_j$ acts. Then the verifier
measures the POVM $\{H_j, I-H_j\}$ and rejects when the measurement
result is '$H_j$'. It is easy to see that the verifier rejects with
probability
\begin{equation*}
  \frac{1}{m}\sum_{j=1}^m\ip{H_j}{\rho} = \frac{1}{m} \ip{H}{\rho}.
\end{equation*}

In the second type of energy measurement, we only measure Pauli
operators. Let $\P_{XZ}$ be the set of the $3^k$ $k$-fold tensor
products of $I,X,Z$ operators. For $XZ$-form Hamiltonians, it suffices
to measure the Pauli operators in $\P_{XZ}$ only. Expand each term
\begin{equation}
  \label{eq:H_expand}
  H_j = \sum_{P\in \P_{XZ}} \alpha_{j,P} P.
\end{equation}
Computing the trace of squared operators on both sides of
Eq.~\eqref{eq:H_expand}, we have
\begin{equation*}
  \sum_{P\in \P_{XZ}} \alpha_{j,P}^2 \le 1.
\end{equation*}
The verifier randomly selects $j\in [m]$ and gets the $k$-qubit state
$\rho_j$ on which $H_j$ acts. He then chooses $P$ uniformly at random
and measures $P$ on $\rho_j$. The verifier rejects with probability
$\abs{\alpha_{j,P}}$ if either $\alpha_{j,P}>0$ and the measurement
result is $+1$, or $\alpha_{j,P}<0$ and the measurement result is
$-1$. The probability of rejection is computed as
\begin{equation*}
  \frac{1}{3^km}\sum_j\sum_P \frac{\abs{\alpha_{j,P}} +
    \alpha_{j,P}\ip{P}{\rho_j}}{2} = \frac{\alpha m +
    \ip{H}{\rho}}{2\cdot 3^km},
\end{equation*}
where
\begin{equation}
  \label{eq:alpha}
  \alpha = \sum_{j,P}\abs{\alpha_{j,P}}/m
\end{equation}
is a constant determined by the Hamiltonian.

We note that the second type of the energy measurement is less
efficient but the probabilities of rejection in these two settings are
linearly related. In fact, it is easy to see that the rejection
probability in the first setting is $p$ if and only if the rejection
probability of the second setting is
\begin{equation*}
  \frac{\alpha+p}{2\cdot 3^k}.
\end{equation*}

We now give the non-local game for the local Hamiltonian problem as in
Fig.~\ref{fig:lhgame}. To measure the energy, we further assume that
the stabilizer $\Stab$ used in the game has the logical $X,Z$
operators $L_X$ and $L_Z$ that are products of $I$, $X$, $Z$. This is
the case for both the five-qubit code and the four-qubit quantum error
detecting code.

\begin{figure}[!htb]
  \begin{framed}
    \ul{Non-Local Game for The Local Hamiltonian Problem}\\[1em]
    Let $\Stab$ be a non-trivial $r$-qubit stabilizer code with
    $XZ$-form generators that encodes at least one qubit and has a
    pair of logical $L_X$, $L_Z$ operators of $XZ$-form. Define two
    question vectors $w_X = (w_{X,i})$ and $w_Z = (w_{Z,i})$ as
    follows. The entry $w_{D,i}$ is $*$, $0$, or $1$, if the $i$-th
    Pauli factor of the logical operator $L_D$ is $I$, $X$, $Z$
    respectively for $D=X,Z$. For an $XZ$-form, $k$-local Hamiltonian
    problem $(H,a,b)$, and a small probability $p$ chosen later, we
    consider the following multi-player non-local game. It involves a
    classical verifier and $r$ players $(i)$ for $i\in[r]$. The
    verifier performs the first test with probability $p$, and the
    second test with probability $1-p$:
    \begin{enumerate}\setlength{\itemsep}{0pt}
    \item {\it Energy Check.} Select $j\in [m]$ uniformly at random.
      Expand $H_j$ as
      \begin{equation*}
        H_j = \sum_{P\in \P_{XZ}} \alpha_{j,P} P,
      \end{equation*}
      and let $J \subset [n]$ be the set of $k$ qubits $H_j$ acts on
      non-trivially. Select an operator $P$ in $\P_{XZ}$ uniformly at
      random. For each $u\in J$, define $w_{u,i}$ to be $w_{X,i}$ or
      $w_{Z,i}$ if the tensor factor in $P$ acting on qubit $u$ is $X$
      or $Z$, respectively. Define $q^{(i)}_u = (u,w_{u,i})$ and
      $\vq = \bigl( q^{(i)}_u \bigr)_{u\in J}$. Send the question
      $\vq$ to the $r$ players. Receive a $k$-bit answer
      $a^{(i)} = \bigl( a^{(i)}_u \bigr)$ from each player. The
      verifier rejects with probability $\abs{\alpha_{j,P}}$ if either
      the parity of the answer bits not corresponding to the null
      question is even and $\alpha_{j,P}>0$, or the parity is odd and
      $\alpha_{j,P}<0$.
    \item {\it Encoding Check.} Play the $(k,n)$-stabilizer game.
    \end{enumerate}
    \caption{The non-local game for local Hamiltonian problems}
    \label{fig:lhgame}
  \end{framed}
\end{figure}

\begin{theorem}
  \label{thm:lhgame}
  There exists constants $C$ and $\kappa$ such that the following
  holds. If $p$ in the $r$-player game for an $XZ$-form, $k$-local
  Hamiltonian problem in Fig.~\ref{fig:lhgame} is chosen to be
  $C(b-a)^\kappa/n^\kappa$, and $\alpha$ is defined as in
  Eq.~\eqref{eq:alpha}, then
  \begin{enumerate}
  \item For yes-instance of the $k$-local Hamiltonian problem, the
    non-local value of the game is at least
    \begin{equation*}
      (1-p)\omega^*_{\text{S}} + p \biggl(1-\frac{\alpha+a}{2\cdot
        3^k}\biggr).
    \end{equation*}
  \item For no-instance of the $k$-local Hamiltonian problem, the
    non-local value of the game is at most
    \begin{equation*}
      (1-p)\omega^*_{\text{S}} + p \biggl(1-\frac{\alpha+(a+b)/2}{2\cdot
        3^k}\biggr).
    \end{equation*}
  \end{enumerate}
\end{theorem}

\begin{proof}
  First consider the completeness of the game. If the local
  Hamiltonian problem is a yes-instance, there exists a quantum
  witness state $\ket{\psi} \in \B^{\otimes n}$ such that
  \begin{equation*}
    \bra{\psi} H \ket{\psi} \le am.
  \end{equation*}
  We construct the strategy for the $r$ players as follows. For each
  qubit $u$ of $\ket{\psi}$, we encode it with the stabilizer code
  $\Stab$ and let player $(i)$ hold the $i$-th encoded qubit of $u$.
  When receiving the question $(u,w)$ from the verifier, the players
  measure their share of qubit $u$ with $X$, $Z$, $X'$, $Z'$
  correspondingly if $w=0,1,2,3$.

  For this strategy, the players can win the {\it Encoding Check\/}
  part with optimal probability $\omega^*_{\text{S}}$. In the {\it
    Energy Check\/} part of the game, the measurement of the logical
  $X$ and logical $Z$ is essentially an implementation of the second
  type of energy measurement on the state $\psi$. The rejection
  probability in this part is
  \begin{equation*}
    \frac{\alpha m + \ip{H}{\op{\psi}{\psi}}}{2\cdot 3^km}.
  \end{equation*}
  Therefore, the acceptance probability of the game $\omega^*$ is at
  least
  \begin{equation*}
    (1-p)\omega^*_{\text{S}} + p \biggl(1-\frac{\alpha+a}{2\cdot
      3^k}\biggr).
  \end{equation*}

  Next, we prove the soundness of the game. If the local Hamiltonian
  problem defined by $(H,a,b)$ is a no-instance, we need to prove an
  upper bound of the non-local value of the game.

  Consider any strategy $\S$ that has acceptance probability
  $\omega^*_{\text{S}} - \epsilon$ in the {\it Encoding Check\/} part
  of the game. Theorem~\ref{thm:mqsgame} states that there are
  isometries $V_i \in \Lin(\H_i, \B^{\otimes n} \otimes \hat{\H}_i)$,
  measurements $N_\vq$ associated with
  $\bigl\{ V_i^\dagger D_{q_s} V_i \bigr\}_{s=1}^k$, such that
  \begin{equation*}
    d_\rho (M_\vq , N_\vq) \le O(n^\kappa \epsilon^{1/\kappa}).
  \end{equation*}

  Consider the strategy $\tilde{\S} = (\rho, \{R_q\},\{N_\vq\})$. The
  value of $\S$ and $\tilde{\S}$ for the {\it Energy Check\/} differ
  at most by $O(n^\kappa \epsilon^{1/\kappa})$ by
  Lemma~\ref{lem:drho}.

  Strategy $\tilde{\S}$ uses honest $X$, $Z$ measurement on the
  logical space of the error correcting code and we claim that it must
  have value at most
  \begin{equation*}
    1 - \frac{\alpha+b}{2\cdot 3^k},
  \end{equation*}
  in the {\it Energy Check\/} part.
  Otherwise, the state of the first $rn$ qubits after the application of
  $\bigotimes_i V_i$ has rejection probability at most
  $(\alpha+b)/(2\cdot 3^k)$ in the energy measurement using logical
  $X$, $Z$ operators $L_X$, $L_Z$. This implies the existence of
  $n$-qubit state that has rejection probability at most
  $(\alpha+b)/(2\cdot 3^k)$ in the second type energy measurement,
  which is a contradiction to the no-instance condition of the local
  Hamiltonian problem.

  Therefore, the value of strategy $\S$ is at most
  \begin{equation}
    \label{eq:lh1}
    (1-p) (\omega^*_{\text{S}} - \epsilon) + p \bigl(
    1-\frac{\alpha+b}{2 \cdot 3^k} + cn^\kappa \epsilon^{1/\kappa}
    \bigr),
  \end{equation}
  for constants $c, \kappa$ large enough.

  Maximizing the expression as a function of $\epsilon$, it is easy to
  see that the maximum value is achieved at
  \begin{equation*}
    \epsilon = \Bigl( \frac{pcn^\kappa}{(1-p)\kappa}
    \Bigr)^{\kappa/(\kappa-1)}.
  \end{equation*}
  Substituting this into Eq.~\eqref{eq:lh1}, $\omega^*(\S)$ is upper
  bounded by
  \begin{equation*}
    (1-p)\omega^*_{\text{S}} + p \bigl( 1-\frac{\alpha+b}{2\cdot 3^k} +
    \Delta \bigr),
  \end{equation*}
  for
  \begin{equation*}
    \Delta = \bigl( 1-\frac{1}{\kappa} \bigr) cn^\kappa
    \Bigl(
    \frac{pcn^\kappa}{(1-p)\kappa}
    \Bigr)^{1/(\kappa-1)}.
  \end{equation*}
  Choosing $\kappa$ large and $p$ small such that $(1-p)\kappa \ge 1$,
  we can bound $\Delta$ as
  \begin{equation*}
    \Delta \le cn^\kappa (pcn^\kappa)^{1/(\kappa-1)} = (pc^\kappa
    n^{\kappa^2})^{1/(\kappa-1)}.
  \end{equation*}
  Finally, if we choose a constant $C$ small enough and
  \begin{equation*}
    p = C (b-a)^{\kappa-1} / n^{\kappa^2},
  \end{equation*}
  we have
  \begin{equation*}
    \Delta \le \frac{b-a}{4 \cdot 3^k},
  \end{equation*}
  and
  \begin{equation*}
    \omega^*(\S) \le (1-p)\omega^*_{\text{S}} +
    p(1-\frac{\alpha+(a+b)/2}{2\cdot 3^k}).
  \end{equation*}
  This concludes the proof of the theorem by choosing $\kappa$ large
  enough.
\end{proof}

Finally, Theorem~\ref{thm:main} follows by using the stabilizer for
the four-qubit error detecting code in the game and noticing that the
completeness and soundness has inverse polynomial gap in
Theorem~\ref{thm:lhgame}.

\section*{Acknowledgments}

The author acknowledges helpful discussions with Richard Cleve, Debbie
Leung, Fang Song, Thomas Vidick, Guoming Wang, John Watrous, Xiaodi Wu
and Bei Zeng on related problems. This work is supported by NSERC and
ARO.

\printbibliography


\end{document}